\documentclass[11pt,draftcls,onecolumn]{IEEEtran}

\usepackage{amsbsy,color,multicol}
\usepackage{floatflt} 

\usepackage{amsmath}
\usepackage{amssymb}
\usepackage{times}
\usepackage{graphicx}
\usepackage{xspace}
\usepackage{paralist} 
\usepackage{setspace} 
\usepackage{xypic}
\xyoption{curve}
\usepackage{latexsym}
\usepackage{theorem}
\usepackage{ifthen}
\usepackage{caption}
\usepackage{subcaption}
\usepackage{makeidx,multirow,textcomp,longtable,afterpage,setspace}
\usepackage[english]{babel}
\usepackage[latin1]{inputenc}
\usepackage[ruled,vlined]{algorithm2e}
\usepackage{booktabs}
\usepackage{mathtools}

\newlength\myindent
\newlength\mycolwid

\newcounter{proofcount}
\newenvironment{proof}
{
\stepcounter{proofcount} {\textit{Proof}:} 
}

{\theoremheaderfont{\it} \theorembodyfont{\rmfamily}
\newtheorem{theorem}{Theorem}

\newtheorem{definition}[theorem]{Definition}

}

\title{Emergent Behavior in Multipartite Large Networks: Multi-virus Epidemics}

\author{Augusto Santos$^{*}$, Jos\'e M.~F.~Moura$^{\natural}$, and Jo\~{a}o Xavier$^\dagger$
\thanks{$*$ A.~Santos is with the Dep.~of Electrical and Computer Engineering, Carnegie Mellon University, USA, and Instituto de Sistemas e Robotica (ISR), Instituto Superior T\'{e}cnico (IST), Av.~Rovisco Pais, Lisboa, Portugal (augustos@andrew.cmu.edu).}
\thanks{$^\natural$J.~M.~F.~Moura is with the Dep.~of Electrical and Computer Engineering,
Carnegie Mellon University, Pittsburgh, PA, USA 15213, ph:(412)268-6341, fax: (412)268-3890 (moura@ece.cmu.edu).}
\thanks{$\dagger$ J.~Xavier is with ISR, IST, Av.~Rovisco Pais, Lisboa, Portugal (jxavier@isr.ist.utl.pt).}
\thanks{The work of Jos\'e M.~F.~Moura and Augusto Santos was supported in part by the National Science Foundation under Grant \#~CCF--1011903 and in part by the Air Force Office of Scientific Research under Grant \#~FA--95501010291.

The work of J.~Xavier and A.~Santos was also supported by the Funda\c{c}\~{a}o para a Ci\^{e}ncia e a Tecnologia (FCT) (Portuguese Foundation for Science and Technology) through the Carnegie Mellon$|$Portugal Program under Grant SFRH/BD/33516/2008, CMU-PT/SIA/0026/2009 and SFRH/BD/33518/2008, and by ISR/IST pluriannual funding (POSC program, FEDER).}}

\begin{document}
\maketitle \thispagestyle{empty} \maketitle



\begin{abstract}

Epidemics in large complete networks is well established. In contrast, we consider epidemics in non-complete networks. We establish the fluid limit \emph{macroscopic} dynamics of a multi-virus spread over a multipartite network as the number of nodes at each partite or \emph{island} grows large. The virus spread follows a peer-to-peer random rule of infection in line with the Harris contact process (refer to~\cite{Metastability}). The model conforms to an SIS (susceptible-infected-susceptible) type, where a node is either infected or it is healthy and prone to be infected. The local (at node level) random infection model induces the emergence of structured dynamics at the \emph{macroscale}. Namely, we prove that, as the multipartite network grows large, the normalized Markov jump vector process $\left(\overline{\mathbf{Y}}^\mathbf{N}(t)\right)=\left(\overline{Y}_1^\mathbf{N}(t),\ldots, \overline{Y}_M^\mathbf{N}(t)\right)$ collecting the fraction of infected nodes at each island $i=1,\ldots,M$, converges weakly (with respect to the Skorokhod topology on the space of \emph{c\`{a}dl\`{a}g} sample paths) to the solution of an $M$-dimensional vector nonlinear coupled ordinary differential equation. In the case of multi-virus diffusion with $K\in\mathbb{N}$ distinct strains of virus, the Markov jurmp matrix process $\left(\overline{\mathbf{Y}}^\mathbf{N}(t)\right)$, stacking the fraction of nodes infected with virus type~$j$, $j=1,\ldots,K$,  at each island $i=1,\ldots,M$, converges weakly as well to the solution of a $\left(K\times M\right)$-dimensional vector differential equation that is also characterized.
\end{abstract}








\newpage
\section{Introduction}

Many complex dynamical systems exhibit \emph{emergent behavior} -- a well-structured macroscopic dynamics induced by
simple, possibly random, local rules of interacting agents. Flocks of birds, ant colonies, beehives, brain neural networks, invasive tumor growth, and epidemics are all examples of large scale interacting agents systems displaying complex adaptive functional behaviors. Under appropriate initial conditions, a flock of birds reaches consensus on its cruise velocity
while each bird probes only its nearest neighbors dynamics without a preferred leader in the flock (refer to~\cite{Flock}). This gives rise to synchronized flocking flying formations. Ant colonies can design optimal trails to access sources of food even though no ant bears the cognitive ability to shape up the colony to its blueprint mature optimal global behavior. Roughly, each scout-ant wanders around randomly tracking the leftover pheromone released by its scout peers. Reference~\cite{ants} establishes the emergent dynamics of an idealized stochastic network model for ant colonies as the fluid limit dynamics of the network model (as the colony grows large). Seizure is an intricate outcome of the complex neural network dynamics of the brain. Reference~\cite{seizures} presents an overview of graphical dynamical models that have been applied to better understand the nature of seizures and bridge the microscopical electrical activity in the brain with the clinical observations of the phenomenon.

The challenge in studying such large scale systems lies in their \textbf{high-dimensionality} plus the \textbf{coupling} among the agents via their interactions. Together, these are the needed ingredients to induce emergent behavior. For instance, consider $N$ agents whose state-vector $\mathbf{X}^{N}(t):=\left(X_1(t),\ldots,X_N(t)\right)$ evolves as a jump Markov process over the state space $\mathcal{S}^{N}:=\left\{0,1\right\}^N$. If the agents are independent, then it turns out that the state of each agent evolves as a jump Markov process and, moreover, any state construct $\left(f\left(X_1(t),\ldots,X_N(t)\right)\right)$, where $f:\left\{0,1\right\}^N \rightarrow \mathbb{R}^M$ bears appropriate measurability properties (we skip the details here), is a Markov jump process. For instance, the fraction of agents at state $1$, $f\left(X_1(t),\ldots,X_N(t)\right)=\sum_{i=1}^N X_i(t)/N$, is Markov. Even for large $N$, due to the independence assumption, a qualitative analysis of $\left(X^{N}(t)\right)$ becomes tractable, but, in this example of independent agents, any weak law of large numbers will reflect the average behavior of each individual agent rather than an emergent global cooperative behavior. When the agents are coupled -- e.g., an agent switches to state $1$ with a rate that is proportional to the number of its neighbors in state~$1$--then, in general, neither the state of each agent is Markov nor the \emph{macroscopic} low-dimensional states $\left(f\left(X_1(t),\ldots,X_N(t)\right)\right)$ are Markov and studying the \emph{microscopic} high-dimensional dynamical system $\left(\mathbf{X}^N(t)\right)$ becomes quickly unfeasible with the number of agents $N$. Establishing the emergent dynamics or, in other words, the functional weak law of large numbers under an arbitrary coupling topology of the agents is challenging. For the special case of a \emph{complete} topology of interaction -- any agent evenly affects any other agent in the cloud, -- low-dimensional \emph{macroscopic} state-variables may still be Markov, even though the state of each individual agent is no longer Markov. Again, for complete networks, the fraction of infected nodes $f\left(X_1(t),\ldots,X_N(t)\right)=\sum_{i=1}^N X_i(t)/N$ is Markov. Under this complete network setting, the emergent behavior is framed as the fluid limit dynamics of a global state variable $\left(\mathbf{Y}(t)\right):=\left(f\left(X_1(t),\ldots,X_N(t)\right)\right)$ of interest. For example, reference~\cite{Antunes2} considers a multiclass flow of packets over a \textbf{complete} network with finite capacity nodes. It defines the macroscopic state variable $\left(\mathbf{Y}^{N}(t)\right)=\left(Y_1^N(t),\ldots,Y^N_L(t)\right)$ that collects the fraction of nodes $Y_i^N(t)$ with a particular distribution $i$ of packets over the different classes. Reference~\cite{Antunes2} proves that the empirical distribution $\left(\mathbf{Y}^{N}(t)\right)$ converges weakly, with respect to the Skorokhod topology on the space of sample paths, to the solution of a vector ordinary differential equation.

For general topologies, the evolution of \emph{macroscopic} state variables is intricately tied to the high-dimensional \emph{microscopic} state $\left(\mathbf{X}^{N}(t)\right)$ of the system. Reference~\cite{4549746} proposes to consider the impact of the topology on the diffusion of a virus in the network, but, to overcome the coupling difficulty that arises with non complete networks, reference~\cite{4549746} departs from a peer-to-peer diffusion model. The authors in~\cite{4549746} replace the exact transition rates of the microstate process $\left(\mathbf{X}(t)\right)$ by their average to establish their $N$-intertwined model. Were the states of the nodes independent processes (a very strong assumption) and the resulting $N$-intertwined model would be an exact model to describe the dynamics of the likelihood of infection of each node as pointed out by the authors.


In this paper, we go beyond the complete network model to establish the \textbf{exact} meanfield dynamics of a multi-virus epidemics over the class of multipartite networks, without making any artificial simplifying assumptions. We assume a stochastic network model for the peer-to-peer spread of different strains of virus among a cloud of agents to establish the emergent dynamics of the epidemics. The emergent behavior is the fluid limit dynamics of the fraction of infected nodes over time. Namely, we show that, when the number of
agents goes to infinity in a way to be described momentarily, the fraction of infected agents at each \emph{island} in the multipartite network converges weakly to the solution of a set of nonlinear ordinary differential equations.

   We briefly outline the paper. In Section~\ref{sec:problemformulation}, we set the problem formulation, defining the peer-to-peer stochastic network model underlying the microscopic dynamics of the diffusion of the strains of virus. In Section~\ref{sec:meanfield}, we establish the meanfield dynamics for a single virus spread over a bipartite network, that is, we prove that $\left(\overline{\mathbf{Y}}^\mathbf{N}(t)\right)=\left(\overline{Y}_1^\mathbf{N}(t), \overline{Y}_2^\mathbf{N}(t)\right)$ collecting the fraction of infected nodes at each island converges weakly to the solution of an ordinary differential equation. In Section~\ref{sec:extension}, we extend the meanfield proof in Section~\ref{sec:meanfield} to the general case of a multipartite network under multi-virus spread. Finally, in Section~\ref{sec:conclusion}, we conclude the paper.

\section{Problem Formulation}\label{sec:problemformulation}

In this Section, we formally introduce the stochastic network process that models the peer-to-peer diffusion of multiple strains of virus over multipartite networks. First, we set up the environment where the epidemics takes place. Let $G=\left(V,E\right)$ be an undirected network, where $V=\left\{1,2,\ldots,N\right\}$ and $E=\left\{\left\{i,j\right\}\,:\,i,j\in V\right\}$ represent the set of nodes and edges of the graph $G$, respectively. We say that two nodes $i,j\in V$ are connected and represent it as $i\sim j$, if $\left\{i,j\right\}\in E$. In this paper, we establish the mean field dynamics of a multi-viral strain epidemics over the class of multipartite networks that is defined next.
\begin{definition}[Multipartite network]\label{def:multipartitenetwork}
A network $G=\left(V,E\right)$ is multipartite if there exists a partition $\overline{V}=\left\{V_1,\ldots,V_M\right\}$ of $V$ such that $\left\{a,b\right\}\notin E$ for any $a,\,b\in V_i$ and any $i\in\left\{1,\ldots,M\right\}$. Also,
\begin{align*}
u\in V_i,\,v\in V_j,\,u\sim v\Rightarrow w\sim r,\,\,\,\, \forall\, {w\in V_i,\,r\in V_j}
\end{align*}
with $i\neq j$. When $M=2$, the multipartite network is called \textit{bipartite}.
\hfill$\small \blacksquare$
\end{definition}
The elements $V_i$ of the partition $\overline{V}$ are referred to as \textit{islands}. In words, if two nodes of different islands $V_i$ and $V_j$ are connected then any node from~$V_i$ is connected to any node from~$V_j$. In this case, we say that islands $V_i$ and $V_j$ are connected and refer to it as $V_i\sim V_j$. This allows us to abstract the supertopology structure of islands as illustrated in Figure~\ref{fig:multipartite}.
\begin{figure} [hbt]
\begin{center}
\includegraphics[scale= 0.5]{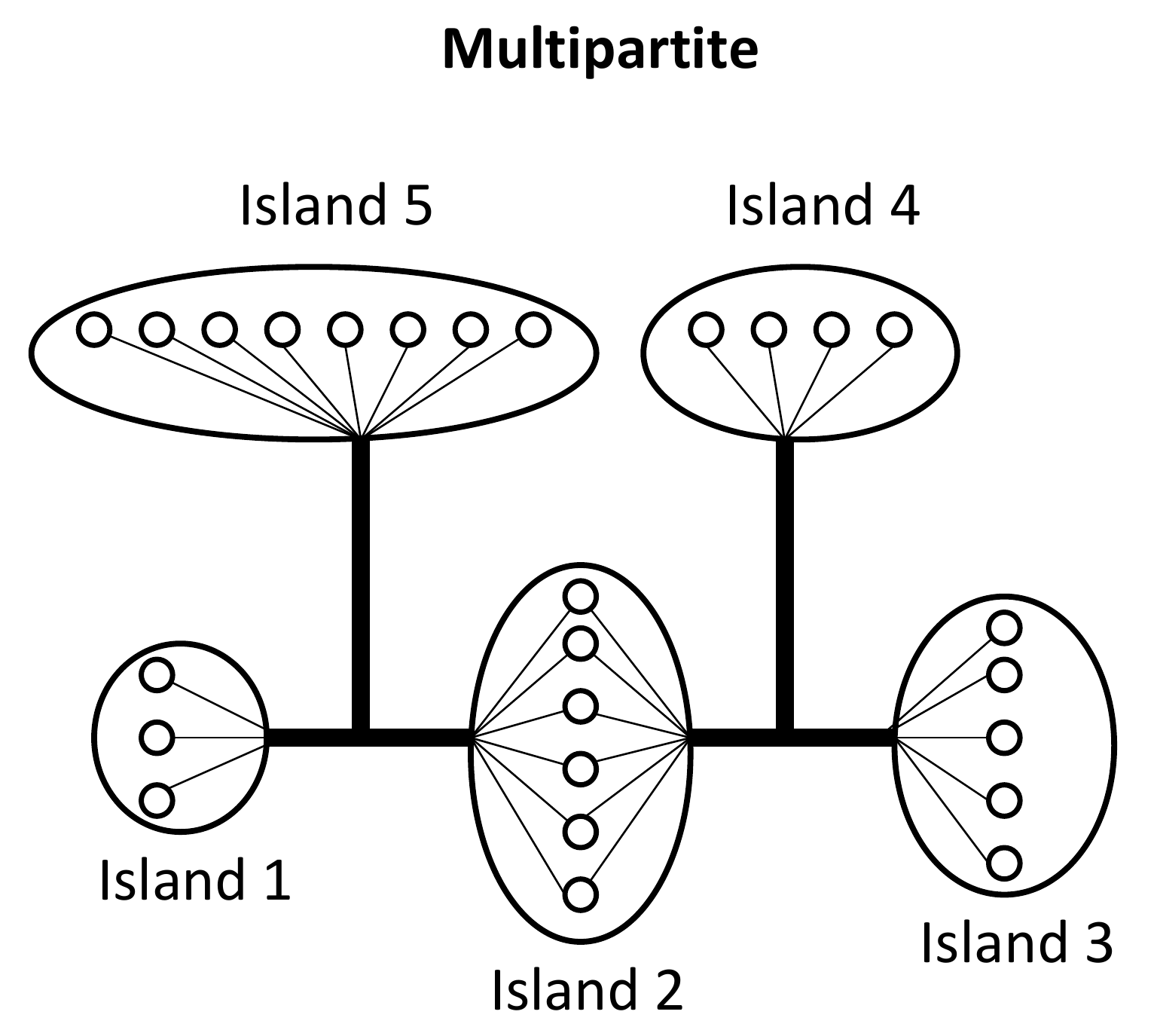}
\caption{Multipartite network representation. Nodes from the same island cannot transmit the infection amongst themselves. Nodes from an island can transmit the virus to nodes in neighboring islands. For instance, any node from island $5$ can infect any node from islands $1$ and $2$.}\label{fig:multipartite}
\end{center}
\end{figure}
Also, we refer to
\begin{equation}
\mathcal{N}\left(V_i\right)=\left\{V_j\,:\,V_j\sim V_i\right\}\nonumber
\end{equation}
as the superneighborhood of island $V_i$ and $d_i=\left|\mathcal{N}\left(V_i\right)\right|$ refers to the superdegree or number of neighboring islands of island $V_i$. As an example, for the superneighbors of island $1$, in Figure~\ref{fig:multipartite}, we have $\mathcal{N}\left(1\right)=\left\{2,5\right\}$ and thus, $d_1=2$.

Given a graph $G=\left(E,V\right)$, we define the sequence of induced multipartite networks $G^{\mathbf{N}}=\left(E^{\mathbf{N}},V^{\mathbf{N}}\right)$, indexed by $\mathbf{N}=\left(N_1,\ldots,N_M\right)\in\mathbb{R}^{M}$, where $M$ is the fixed number of islands, $N_i$ is the number of nodes at the $i$th island of $G^{\mathbf{N}}$ and
\begin{equation}
V^{\mathbf{N}}_i\sim V^{\mathbf{N}}_j \Leftrightarrow \ell_i\in V^{\mathbf{N}}_i\sim \ell_j\in V^{\mathbf{N}}_j
\end{equation}
for all $\ell_i,\ell_j\in V$, where $V^{\mathbf{N}}_i$ is the $i$th island of $G^{\mathbf{N}}$, as depicted in figure~\ref{fig:multipartite2}. In words, all the multipartite graphs in the sequence $G^{\mathbf{N}}$ share the same supertopology imposed by the topology of $G$, differing only on the number of nodes per island that is given by the upper-index $\mathbf{N}$. Given a graph $G$, we are interested in obtaining the limiting dynamics of the fraction of infected nodes per island and per strain over $G^{\mathbf{N}}$ as $\mathbf{N}$ grows to infinity.
\begin{figure} [hbt]
\begin{center}
\includegraphics[scale= 0.5]{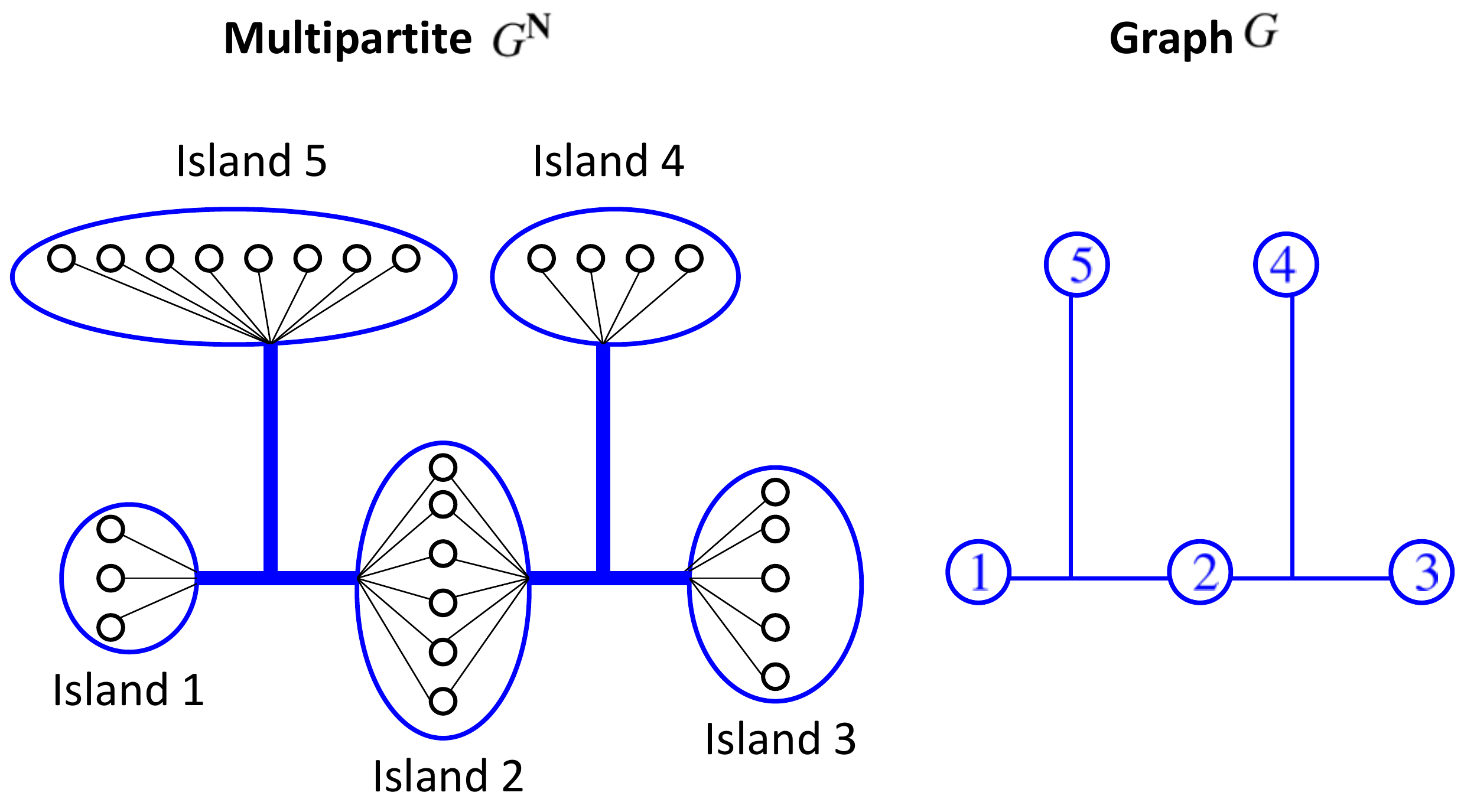}
\caption{Illustration of the multipartite network $G^{\mathbf{N}}$, with $\mathbf{N}=\left(3,6,5,4,8\right)$ and underlying graph $G$.}\label{fig:multipartite2}
\end{center}
\end{figure}
%
%

Now, we introduce the dynamical model of diffusion. A node in a $G^\mathbf{N}$ multipartite network may assume one of two possible states: infected or healthy and susceptible to infections. We define the binary \emph{tensor} microprocess $\left(\mathbf{X}^\mathbf{N}(t)\right)$ as conveying the state of each node over time in $G^{\mathbf{N}}$. By $X_{iIk}^{\mathbf{N}}(t)=1$, we refer to node $i$ at island $I$ being infected at time $t$, $t\geq 0$, with virus strain $k\in\left\{1,2,\ldots,K\right\}$, and by $X_{iIk}^{\mathbf{N}}(t)=0$ if the corresponding node $i$ is healthy or infected with other strain $l\neq k$; in this latter case, $X_{iIl}^{\mathbf{N}}(t)=1$. The upper-index $\mathbf{N}$ will permeate all relevant stochastic processes constructs to emphasize the underlying multipartite network $G^{\mathbf{N}}$ induced by $G$. If only one strain of virus is present in the network then, for notational simplicity, we suppress the extra-index $k$ and rather write $X_{iI}^{\mathbf{N}}(t)=1$ to represent that node $i$ from island $I$ is infected at time $t$, $t\geq 0$, and $X_{iI}^{\mathbf{N}}(t)=0$ if otherwise.

Our microscopical model of diffusion of the virus is set at node level and goes as follows. If a node $i$ from island $I$ is $y$-infected--i.e., infected with the virus strain $y$--at time $t$, then it heals after an exponentially distributed random time $T_{Iy}^h(i)\sim {\sf Exp}\left(\mu_I^{y}\right)$ whose distribution depends on the island $I$ and on the type of infection $y$. Once a node $i$ in island $I$ is $y$-infected, it transmits the infection to a randomly chosen node at the neighbor island $J\in\mathcal{N}\left(I\right)$ after an exponentially distributed random time~$T_{IJy}^c(i)\sim {\sf Exp}\left(\gamma^{y}_{IJ}\right)$ whose distribution only depends on the ordered pair $\left(I,J\right)\in\left\{1,2,\ldots,M\right\}^2$ of communicating islands $I$ and $J$, and the type of infection $y$. Whenever considering only one strain of virus, we drop the strain subindex $y$. Also, if there is no room for ambiguity, we drop the node identity $i$, writing $T_{Iy}^h$ for the healing time and $T_{IJy}^c$ for the infection time of a node at island $I$. If the chosen node $j$ at island $J$ is already $k$-infected at the time of infection $t$, $t\geq 0$, then nothing happens; that is, $X_{jJk}^{\mathbf{N}}(t)=1$. Therefore, $\sum_{k}X_{iIk}(t)\leq 1$ for all $i\in I$, $I\in\left\{1,\ldots,M\right\}$ and $t$, $t\geq 0$, or in words, a node can only be infected by one strain of virus at a time.

To summarize, an infected node $i \in I$ activates $d_I+1$ independent exponentially distributed random variables, where $d_I=\left|\mathcal{N}\left(I\right)\right|$ of them are associated to the times for infection and one to the time for healing. Each of the $d_I$ random variables for infection is dedicated to one superneighbor $J\in \mathcal{N}\left(I\right)$ of island $I$ as illustrated in Figure~\ref{fig:dedicated}.
\begin{figure} [hbt]
\begin{center}
\includegraphics[scale= 0.5]{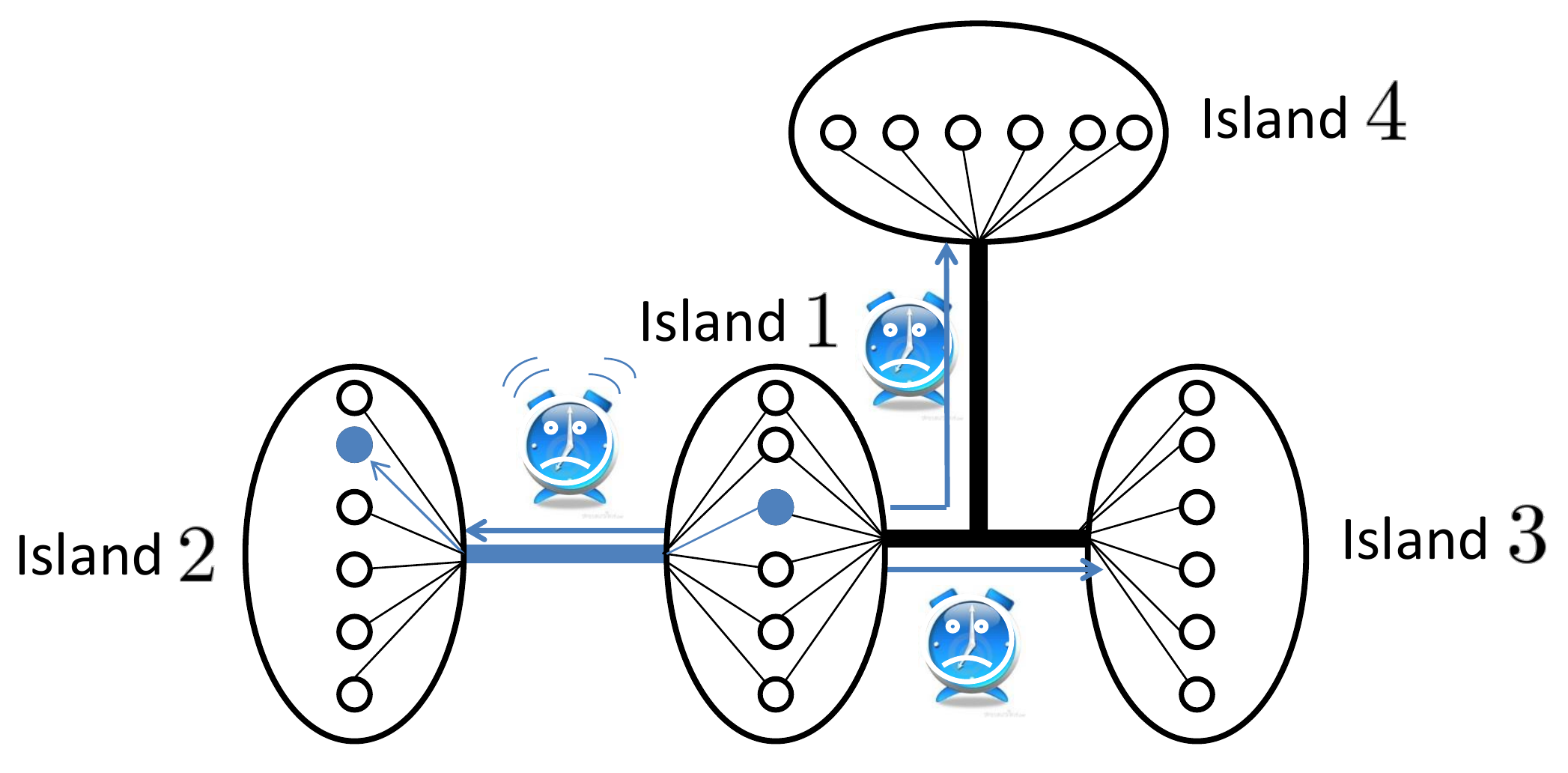}
\caption{Illustration of an infection. The infected (dark colored) node at island $1$ activates three exponentially distributed random \emph{clocks}, each dedicated to one neighbor island. The clock dedicated to island $2$ rings after a time $T_{12y}^c\sim {\sf Exp}\left(\gamma^{y}_{12}\right)$ and a node from island $2$ is randomly picked to be infected. Also, the infected node in island $1$ heals after a time $T_{1y}^h\sim {\sf Exp}\left(\mu^{y}_{1}\right)$. }\label{fig:dedicated}
\end{center}
\end{figure}
As an example, if a node from island $1$ in Figure~\ref{fig:dedicated} is $y$-infected, then, after a time interval of length $T_{12y}^c\sim {\sf Exp}\left(\gamma^{y}_{12}\right)$, it picks randomly a node from island $2$ and infects it. Also, after a time $T_{13y}^c\sim {\sf Exp}\left(\gamma^{y}_{13}\right)$, that is independent of $T_{12y}^c$ and the healing time $T_{1y}^h$, it chooses randomly a node from island $3$ and infects it as long as it is still infected.

The microscopic process $\left(\mathbf{X}^\mathbf{N}(t)\right)$ thus, evolves through jumps according to the triggering of a sequence of independent exponentially distributed random variables. All time service random variables are assumed to be independent and have support in a single probability space $\left(\Omega, \mathcal{F}, \mathbb{P}\right)$. We denote by $\left(\mathcal{F}^\mathbf{N}_t\right)_{t\geq 0}$ the natural filtration induced by the sequence of independent random variables. That is, $\mathcal{F}^\mathbf{N}_t$ gives us information on the values of all random variable times (healing or infection) involved in the evolution of $\left(\mathbf{X}^\mathbf{N}(t)\right)$ up to time $t$, $t\geq 0$. Note that by construction $\left(\mathbf{X}^\mathbf{N}(t)\right)$ is adapted to $\left(\mathcal{F}^\mathbf{N}_t\right)_{t\geq 0}$, i.e., $\sigma\left\{\mathbf{X}^\mathbf{N}(s)\,:\,0\leq s\leq t\right\}\subset\mathcal{F}^\mathbf{N}_t$, for all $t$, $t\geq 0$, where $\sigma\left\{\mathbf{X}^\mathbf{N}(s)\,:\,0\leq s\leq t\right\}$ represents the natural filtration of the process $\left(\mathbf{X}^\mathbf{N}(t)\right)$.

Analyzing the full-microstate of the network over time according to the local infection model just presented becomes quickly unfeasible with the number of nodes in the network as the microscopic process $\left(\mathbf{X}^\mathbf{N}(t)\right)$ is high-dimensional. Instead, we are interested in characterizing macroscopically the virus evolution in the multipartite network, namely, studying the dynamics of the number or fraction of infected nodes at each island. To fix ideas, we assume single virus epidemics for the rest of this Section, unless otherwise stated. We refer to $\left(\mathbf{Y}^{\mathbf{N}}(t)\right)=\left(Y_1^\mathbf{N}(t),\ldots,Y_M^\mathbf{N}(t)\right)$ as the macroprocess that stacks the number of infected nodes at each island. The normalized vector process $\left(\overline{\mathbf{Y}}^{\mathbf{N}}(t)\right)=\left(\overline{Y}_1^\mathbf{N}(t),\ldots,\overline{Y}_M^\mathbf{N}(t)\right)$
collects the corresponding fractions of infected nodes per island, where $\overline{Y}_i^\mathbf{N}(t)=Y_i^\mathbf{N}(t)/N_i$ is the fraction of infected nodes at island $i$ at time $t$, $t\geq 0$, with $N_i=|V_i|$, the number of nodes at island~$i$. The sequence of macroprocesses $\left(\mathbf{Y}^\mathbf{N}(t)\right)_{\mathbf{N}}$ is indexed by $\mathbf{N}=\left(N_1,\ldots,N_M\right)$, the vector collecting the cardinality $N_i$ of each island $i$ in the underlying multipartite network $G^\mathbf{N}$. It turns out that, from the microscopic model of peer-to-peer infection previously described, $\left(\mathbf{Y}^{\mathbf{N}}(t)\right)$ is a jump Markov process with transition rate matrix given by
\begin{eqnarray}
Q\left(\mathbf{Y}^\mathbf{N}(t),\mathbf{Y}^\mathbf{N}(t)-e_i\right) & = & \mu_i Y_i^\mathbf{N}(t)\label{eq:rate2}\\
Q\left(\mathbf{Y}^\mathbf{N}(t),\mathbf{Y}^\mathbf{N}(t)+e_i\right) & = & \left(\sum_{j\sim i}\gamma_{ji} Y_j^\mathbf{N}(t)\right)\left(\frac{N_i-Y_i^\mathbf{N}(t)}{N_i}\right),\label{eq:rate1}
\end{eqnarray}
where $e_i\in\mathbb{R}^{M}$ is the canonical vector with the $i$th entry equal to $1$ and the remaining entries equal to $0$, and the lowercase subindexes in equations~(\ref{eq:rate2})-(\ref{eq:rate1}) refer to islands. In equation~(\ref{eq:rate2}), we represent the rate to decrease the infected population at island $i$ by one. This happens once any infected node from island $i$ heals,
\begin{equation}
T^h_i=\min\left\{T_{i}^h(k)\,:\,X^\mathbf{N}_{ki}(t)=1,\,k\in i\right\}\sim {\sf Exp}\left(\mu_{i}Y_i^\mathbf{N}(t)\right).\nonumber
\end{equation}
In equation~(\ref{eq:rate1}), we represent the rate to increase by one the infected population at island $i$. In this case, each neighboring island $j\in \mathcal{N}(i)$ of $i$ will have $Y_j^\mathbf{N}(t)$ infected nodes and, thus, after a time
\begin{equation}
T^c_i=\min\left\{T_{ji}^c(k)\,:\,X^\mathbf{N}_{jk}(t)=1,\,k\in j,\,j\sim i\right\}\sim {\sf Exp}\left(\sum_{j\sim i}\gamma_{ji}Y_j^\mathbf{N}(t)\right),\label{eq:min}
\end{equation}
an attempt of infection will be made by a neighboring node at some neighboring island, where $T_{ji}^c(k)\sim {\sf Exp}\left(\gamma_{ji}\right)$ is the time that the infected node $k\in j$ takes to make an attempt of infection towards a node at island~$i$. The \emph{minimum} in equation~(\ref{eq:min}) runs over all infected nodes in all the neighboring islands of $i$. The rate at which an infection from the neighboring islands takes to strike island $i$ is, thus, $\left(\sum_{j\sim i}\gamma_{ji}^{y}Y_j^\mathbf{N}(t)\right)$. As referred in the microscopic model description, if an infection is transmitted to an already infected node, then the state of the sink node remains unchanged, that is, the infected population does not increase. Therefore, the effective rate of infection will be the rate at which infections arriving at island $i$ hit a healthy node, that is, it is given by the arrival rate $\left(\sum_{j\sim i}\gamma_{ji}^{y}Y_j^\mathbf{N}(t)\right)$ times the probability of hitting a healthy node, which is equal to the fraction of healthy nodes at island $i$ at time $t$, $\left(\frac{N_i-Y^\mathbf{N}_i(t)}{N_i}\right)$, since the chosen \emph{victim} node is drawn uniformly randomly. Note that the topology of the underlying network impacts the increasing rate $Q\left(\mathbf{Y}^\mathbf{N}(t),\mathbf{Y}^\mathbf{N}(t)+e_i\right)$ whereas the decreasing rate $Q\left(\mathbf{Y}^\mathbf{N}(t),\mathbf{Y}^\mathbf{N}(t)-e_i\right)$ only relies on the number of infected nodes at the network at time $t$, regardless of the peer-to-peer connections. Note that two or more events--infection or healing of a node--happens at the same time with probability zero, that is, the evolution of the vector macroprocess $\left(\mathbf{Y}^\mathbf{N}(t)\right)$ is driven almost surely through unit jumps of $e_i$ at each time. The goal of the next Section is to explore the Markov structure of the macroprocess $\left(\mathbf{Y}^\mathbf{N}(t)\right)$ to establish weak convergence as the number of nodes per island $\mathbf{N}$ goes to infinite keeping the underlying graph $G$ (and thus, the number of islands $M$) fixed. Namely, the normalized process $\overline{\mathbf{Y}}^\mathbf{N}(t)$ admits a decomposition in terms of a martingale and a drift term that is a functional built upon the transition rates previously presented,
\begin{equation}
\mathbf{\overline{Y}}^\mathbf{N}(t)=\mathbf{\overline{Y}}^\mathbf{N}(0)+\underbrace{\mathbf{\overline{M}}^\mathbf{N}(t)}_{martingale}+\underbrace{\int_0^t F\left(\mathbf{\overline{Y}}^\mathbf{N}(s-)\right)ds}_{drift}.\nonumber
\end{equation}
This decomposition is also known as Dynkin's formula (refer to~\cite{Diffusion}) and holds for any Markov process. In the next Section, we provide an explicit characterization of the martingale term, and we will establish that it converges weakly to zero as the underlying $G^\mathbf{N}$ multipartite network grows large, and as a result (to be also proved) the vector process will converge weakly to the solution of the deterministic ordinary differential equation
\begin{equation}
\frac{d}{dt}\mathbf{\overline{Y}}(t)=F\left(\mathbf{\overline{Y}}(t)\right),\nonumber
\end{equation}
where the vector field $F$ will be characterized momentarily.
\section{Mean Field -- Bipartite Single Virus}\label{sec:meanfield}
In this Section, we establish for a single virus spread over a bipartite network that the empirical distribution sequence $\left(\mathbf{\overline{\mathbf{Y}}^{\mathbf{N}}}(t)\right)$ converges weakly, as the network grows large, to the solution of a deterministic vector differential equation. By the \emph{network grows large}, we mean that $N_m\rightarrow \infty$ for all $m\in\left\{1,\ldots, M\right\}$ with a finite asymptotic ratio $N_i/N_j\rightarrow \alpha_{ji}<\infty$ between neighboring island sizes, as will be clearer momentarily, keeping the number of islands $M$ fixed. To fix ideas, we consider throughout this Section a single-virus epidemics in a bipartite network. We extend the analysis to the multivirus epidemics over a general multipartite network in Section~\ref{sec:extension}. We remark that the Markov jump process $\left(\mathbf{\overline{Y}^{N}}(t)\right)$ admits a decomposition into a martingale term plus a drift term obtained from transition rates as characterized in the previous Section. In Section~\ref{sec:notation}, we fix needed notation. Section~\ref{sec:pointprocesses} provides a pathwise characterization for the process $\left(\mathbf{Y}^\mathbf{N}(t)\right)$, in particular, the pathwise description of the martingale term $\left(\mathbf{M}^\mathbf{N}(t)\right)$ that will be later important to establish the weak convergence to zero of its normalized counterpart $\left(\mathbf{\overline{M}}^\mathbf{N}(t)\right)=\left(M_1^\mathbf{N}(t)/N_1,\ldots,M_M^\mathbf{N}(t)/N_M\right)$. To prove the weak convergence of the normalized process $\left(\mathbf{\overline{Y}}^{\mathbf{N}}(t)\right)$, we start by showing in Section~\ref{sec:martingalevanishes} that the underlying sequence of martingales $\left(\mathbf{\overline{M}^{N}}(t)\right)$ converges weakly (with respect to the Skorokhod topology) to zero as $N_i$ grows large for all $i=1,\ldots,M$ with $M\in\mathbb{N}$, the number of islands, kept fixed. In Section~\ref{sec:tight}, we show that as a consequence $\left(\mathbf{\overline{Y}^{N}}(t)\right)$ is a tight family indexed by $\mathbf{N}$ whose accumulation points (for the weak convergence) are necessarily given by the solutions of a differential equation. By uniqueness of the resulting differential equation (the vector field is globally-Libpschitz), any convergent subsequence converges to the solution of this differential equation. Therefore, since the limiting process is unique, it follows that the whole sequence converges to the solution of the differential equation.


\subsection{Preliminary Notation}\label{sec:notation}
We briefly present the notation used throughout this Section.

$\bullet$ $Y_i^{\bf{N}}(t)$: number of infected nodes at island $i$ at time $t$, $t\geq 0$. The boldface upperscript $\mathbf{N}=\left(N_1, N_2,\ldots,N_M\right)$ stands for a vector stacking the number of nodes $N_i$ at each island $i$. For the purpose of this Section $M=2$.

$\bullet$ $\overline{Y}_i^{\bf{N}}(t)$: fraction of infected nodes at island $i$ at time $t$, $t\geq 0$, $\overline{Y}_i^{\bf{N}}(t)=Y_i^{\bf{N}}(t)/N_i$.

$\bullet$ $\mathbf{Y}^{\mathbf{N}}(t)$: vector stacking the number of infected nodes at each island $i$ at time $t$, $t\geq 0$, $\mathbf{Y}^{\mathbf{N}}(t)=\left(Y_1^{\bf{N}},\ldots, Y_M^{\bf{N}}(t)\right)$.

$\bullet$ $\overline{\mathbf{Y}}^{\mathbf{N}}(t)$: vector stacking the fraction of infected nodes at each island $i$ at time $t$, $t\geq 0$, $\mathbf{\overline{Y}}^{\mathbf{N}}(t)=\left(\overline{Y}_1^{\bf{N}},\ldots, \overline{Y}_M^{\bf{N}}(t)\right)$.

$\bullet$ $\mathcal{N}_{\alpha}:\mathcal{B}_{\left[0,\infty\right)}\rightarrow \mathbb{N}$: Poisson point process with rate $\alpha$, $\alpha>0$. $\mathcal{N}_{\alpha}\left(A\right)$ counts the number of {\em events} in the Borel-set $A\in\mathcal{B}_{\left[0,\infty\right)}$. We index the elements of a family of independent Poisson point processes by an upperscript $\mathcal{N}^{(i)}_{\alpha}$. More details are presented in the next Subsection~\ref{sec:pointprocesses}.

$\bullet$ $\mathcal{Y}^{\mathbf{N}}_i\subset \mathbb{N}^{\mathbf{N}}$: state-space of the process $\left(\mathbf{Y}^{\mathbf{N}}(t)\right)$, defined by \begin{equation}
\mathcal{Y}^{\mathbf{N}}_i=\left\{\mathbf{Y}=\left(Y_1,\ldots,Y_K\right)\in \mathbb{N}^K\,:\, 0\leq \sum_{k=1}^K Y_k\leq N_i\,\right\},\:\forall i=1,\ldots,M.\label{eq:statespace}
\end{equation}

$\bullet$ $D_{\mathbb{R}}\left[0,T\right]$: space of c\`{a}dl\`{a}g (\emph{continue \`{a} droite, limit\'{e} \`{a} gauche})   functions $f:\left[0,T\right]\rightarrow \mathbb{R}$ endowed with the Skorokhod topology, e.g., \cite{Queues}.

$\bullet$ $\mathcal{B}\left(\mathbb{R}^n\right)$: Borel $\sigma$-algebra over $\mathbb{R}^n$ with the standard topology.

$\bullet$ ${\sf Leb}\left(\cdot\right)\,:\,\mathcal{B}\left(\mathbb{R}\right)\rightarrow \left[0,\infty\right]$: Lebesgue measure over the real line $\mathbb{R}$.

$\bullet$ $\left(\overline{\mathbf{Y}}^{\mathbf{N}}(t)\right)\Rightarrow\left(\overline{\mathbf{Y}}(t)\right)$: stands for $\left(\overline{\mathbf{Y}}^{\mathbf{N}}(t)\right)$ converges weakly, in the Skorokhod topology, to the process $\left(\overline{\mathbf{Y}}(t)\right)$.

\subsection{Pathwise Representation}
\label{sec:pointprocesses}

In this Section, we provide a pathwise characterization for the macroprocess $\left(\mathbf{Y}^\mathbf{N}(t)\right)$ built upon the microscopic diffusion model. We briefly present the relevant definitions regarding point processes over the real line $\mathbb{R}$ that will be the building blocks for the pathwise description of $\left(\mathbf{Y}^\mathbf{N}(t)\right)$. For more details, refer to~\cite{Queues}.

\begin{definition}[Point measure]
$\mu\,:\,\mathcal{B}\left(\mathbb{R}\right)\rightarrow \left[0,\infty\right]$ is a point measure on $\mathbb{R}$ if there exists a sequence $a_n\in\mathbb{R}$, $n\in\mathbb{N}$, so that
\begin{equation}
\mu\left(A\right)=\#\left(A\cap \left\{a_n\right\}_{n=1}^\infty\right)=\sum_{i=1}^{\infty} \mathbf{1}_{\left\{a_i\in A\right\}},\,\,\,\,\forall A\in\mathcal{B}\left(\mathbb{R}\right),\nonumber
\end{equation}
that is, $\mu\left(A\right)$ counts the number of points of the sequence $a_n$ in $A$ for any Borelian set $A\in \mathcal{B}\left(\mathbb{R}\right)$.
\end{definition}

We represent the set of point measures on $\mathbb{R}$ as $M_p\left(\mathbb{R}\right)$. Therefore, to each point measure $\mu\in M_p\left(\mathbb{R}\right)$ there exists an underlying real sequence $\left(a_n\right)_{n\in\mathbb{N}}$. A point measure $\mu$ on $\mathbb{R}$ is called Radon, if each compact interval entails only a finite number of elements of the associated sequence, or in other words the set of accumulation points of $\left(a_n\right)$ is empty. We now define a point process on the real line $\mathbb{R}$.
\begin{definition}[Point process]
$\mathcal{N}$ is a point process if it is a Radon point measure valued random variable,
\begin{eqnarray}
\mathcal{N}\,:\,\Omega & \rightarrow & M_p\left(\mathbb{R}\right)\nonumber\\
\omega & \rightarrow & \mathcal{N}\left(\omega,\cdot\right).\nonumber
\end{eqnarray}
\end{definition}

In words, each realization $\omega\in\Omega$ leads to a sequence of points in the real line (void of accumulation points) that underlies the point measure $\mathcal{N}\left(\omega,\cdot\right)$. For simplicity, we refer to the random measure that counts the number of events in a Borel set as $\mathcal{N}\left(\cdot\right)$. Thus, $\mathcal{N}\left(A\right)$ is the random measure of $A\in\mathcal{B}\left(\mathbb{R}\right)$. We now introduce the definition of Poisson point process that will be central to building the macroscopic process $\left(\mathbf{Y}^N(t)\right)$ from the local (at node level) rules of infection.

\begin{definition}[Poisson point process]\label{def:poisson_pp}
$\mathcal{N}_\gamma$ is a Poisson point process on $\mathbb{R}$ with rate $\gamma>0$ if it satisfies the following two conditions
\begin{enumerate}
  \item {\bf [independent increments]} Given $I_1,\,I_2,\,\ldots,\,I_n\subset \mathbb{R}$, $n$ disjoint intervals in the real line then, $\mathcal{N}_\gamma\left(I_1\right),\,\mathcal{N}_\gamma\left(I_2\right),\,\ldots,\,\mathcal{N}_\gamma\left(I_n\right)$ are independent random variables.
  \item {\bf [increment stationarity]} Let $I_1,I_2\subset\mathbb{R}$ be two intervals. Then,
  \begin{equation}
  L:={\sf Leb}\left(I_1\right)={\sf Leb}\left(I_2\right)\Rightarrow \mathcal{N}_\gamma\left(I_1\right)\stackrel{\textrm d}{\sim} \mathcal{N}_{\gamma}\left(I_2\right),\nonumber
  \end{equation}
  that is, $\mathcal{N}_{\gamma}\left(I_1\right)$ and $\mathcal{N}_{\gamma}\left(I_2\right)$ are Poisson distributed random variables with rate (or intensity) parameter $\gamma L$.
\end{enumerate}
\end{definition}
Definition~\ref{def:poisson_pp} implies that the interarrival \emph{time} interval between the elements of the underlying random sequence of the Poisson point process is exponentially distributed with mean $1/\gamma$ (refer to~\cite{Probmeasurebill}). That is, the random sequence of points underlying the point process is constructed so that the time between events $T=a_{n+1}-a_n\sim {\sf Exp}\left(\gamma\right)$ is exponentially distributed. We may refer to a Poisson point process $\mathcal{N}_\gamma$ with rate $\gamma>0$ as $\gamma$-Poisson process.

If a process~$Z(t)$ counts the number of events up to time~$t$, $t\geq 0$, from a Poisson source with rate~$\gamma$, then,
\begin{equation}
Z(t)=\mathcal{N}_\gamma\left(\left.\left(0,t\right.\right]\right).\nonumber
\end{equation}
As a concrete example, consider a permanently infected node at island~$1$ from Figure~\ref{fig:dedicated}. Then, according to the nearest-neighbor infection model described at the end of Section~\ref{sec:problemformulation}, $Z(t)$ counts the number of infections that arrives at island~$2$ due to this permanently infected source up to time $t$, $t\geq 0$, if we consider $\gamma_{12}=\gamma$. If instead of a single permanently infected node we had two, then the rate would double and the process counting the number of infections arriving at island~$2$ would still be Poisson given by
\begin{equation}
Z(t)=\mathcal{N}_{2\gamma}\left(\left.\left(0,t\right.\right]\right).\nonumber
\end{equation}
If the process $\left(Z(t)\right)$ counts the number of events from a discrete-time varying Poisson source, that is, if the real line $\mathbb{R}$ can be partitioned into intervals $\cup_{i=1}^\infty I_i=\mathbb{R}$ so that at each time interval $I_i$ there is a $\gamma_i$-Poisson source $\mathcal{N}^{\left(i\right)}$ acting then
\begin{equation}\label{eq:poisson}
Z(t)=\sum_{i=1}^\infty \mathcal{N}^{\left(i\right)}_{\gamma_i}\left(I_i\cap\left.\left(0,t\right.\right]\right).
\end{equation}
From equation~(\ref{eq:poisson}) and the fact that $\mathcal{N}^{\left(i\right)}_{\gamma_i}$ are random measures, we have the following integral characterization for the process $\left(Z(t)\right)$
\begin{equation}
Z(t)=\sum_{i\,:I_i\cap\left.\left(0,t\right.\right]\neq \emptyset} \int_{0}^{t}\mathbf{1}_{\left\{s\in I_i\right\}}\mathcal{N}^{\left(i\right)}_{\gamma_i}\left(ds\right),\nonumber
\end{equation}
where the integrals are taken with respect to the respective Poisson random measures $\mathcal{N}^{\left(i\right)}_{\gamma_i}$. That is, for each realization $\omega\in \Omega$, $\mathcal{N}^{\left(i\right)}_{\gamma_i}\left(\omega, \cdot\right)$ is a measure over $\mathbb{R}$, and that is the measure under which the integral is defined. For the virus spread case, we can consider that now the number of infected nodes at island~$1$ changes over time (instead of the static example with permanently infected sources). In this case, the process $\left(Z(t)\right)$ counting the number of infections arriving at island~$2$ due to infected nodes in island~$1$ is given by
\begin{equation}\label{eq:infection_integral2}
Z(t)=\sum_{k=1}^{N_1} \int_{0}^{t}\mathbf{1}_{\left\{Y^\mathbf{N}_1(s-)=k\right\}}\mathcal{N}^{\left(k\right)}_{k\gamma}\left(ds\right),
\end{equation}
where, in this case, $\mathcal{N}^{\left(1\right)}_{\gamma},\,\mathcal{N}^{\left(2\right)}_{2\gamma},\ldots,\,\mathcal{N}^{\left(N_1\right)}_{N_1\gamma}$ are \textbf{independent} Poisson processes.
That is, we partitioned the real line up to time $t$, $\left(0,t\right]=\bigcup_{k=1}^{N_1}\left\{0< s\leq t\,:\,Y_1^\mathbf{N}(s-)=k\right\}$ according to the number of infected nodes $Y_1^\mathbf{N}(s-)=k$ at island~$1$ during the time interval $\left.\left(0,t\right.\right]$. During the time intervals where island $1$ has $k$ infected nodes, $Y_1^\mathbf{N}(s-)=k$, the source of infection that strikes island $2$ is Poisson with rate $k\gamma$. Equation~(\ref{eq:infection_integral2}) represents a sample path characterization for the process $\left(Z(t)\right)$ that counts the number of infections that flow from island~$1$ towards island~$2$ up to time $t$, $t\geq 0$. Also, if each event $a_n$ from a $\gamma$-Poisson source is only counted with probability $p$ then,
\begin{equation}\label{eq:poissonprobability}
Z(t)=\mathcal{N}_{\gamma p}\left(\left.\left(0,t\right.\right]\right).
\end{equation}
Namely, in this case $\left(Z(t)\right)$ counts the number of infections from a permanently infected node from island~$1$ to~$2$ that strikes healthy nodes, assuming that a fraction of $\left(1-p\right)$ of the nodes at the sink island~$2$ are always infected.

We now consider all these effects together to build the sample path characterization of our macroprocess $\left(\mathbf{Y}^\mathbf{N}(t)\right)$. As mentioned, in this Section we look at the bipartite network single virus case. Let $\left(I_1(t)\right)$ and $\left(H_1(t)\right)$ be the processes counting the number of nodes that are or were infected (at least once) and number of healings, respectively, up to time $t$, $t\geq 0$, at island~$1$. We have
\begin{eqnarray}
Y_1^N(t)-Y_1^N(0) & = & I_1(t)- H_1(t)\\
&=& \sum_{\ell=1}^{N_1}\sum_{q=1}^{N_2} \mathcal{N}^{\left(\ell,q\right)}_{\gamma_{21} q\left(\frac{N_1-\ell}{N_1}\right)}\left(\left\{0\leq s\leq t\,:\,Y^N_1(s)=\ell,Y^N_2(s)=q\right\}\right)\\
& &-\sum_{\ell=1}^{N_1}\mathcal{N}^{\left(\ell\right)}_{\mu \ell}\left(\left\{0\leq s\leq t\,:\,Y^N_1(s)=\ell\right\}\right)\nonumber\\
& = &  \sum_{\ell= 1}^{N_1}\sum_{q= 1}^{N_2} \int_{0}^{t} {\bf 1}_{\left\{Y_{1}^{\mathbf{N}}(s-)=\ell, Y_{2}^{\mathbf{N}}(s-)=q\right\}}\left(\mathcal{N}^{\left(\ell,q\right)}_{\gamma_{21} q\left(\frac{N_1-\ell}{N_1}\right)}(ds)\right)\label{eq:bipintegral}\\
& &-\sum_{\ell= 1}^{N_1}\int_{0}^{t} {\bf 1}_{\left\{Y_{1}^{\mathbf{N}}(s-)=\ell\right\}}\left(\mathcal{N}^{\left(\ell\right)}_{\mu \ell}(ds)\right).\nonumber
\end{eqnarray}
Note that the difference between the process $\left(I_1(t)\right)$ and the process $\left(Z(t)\right)$ in equation~(\ref{eq:infection_integral2}) is that $\left(Z(t)\right)$ counts the number of arrival infections at island~$1$, which might also hit already infected nodes, whereas $\left(I_1(t)\right)$ counts the number of effective infections that hits healthy nodes and thus increases the infected population. In the latter case, one has to account for the effect described in equation~(\ref{eq:poissonprobability}), where only a fraction $\frac{N_1-k_1}{N_1}$ of nodes at island~$1$ is healthy at time $t$, if $Y_1^\mathbf{N}(t)=k_1$. For a general multipartite network, the pathwise dynamics is given by
\begin{eqnarray}
 \label{eq:stochastic33}
Y_{i}^{\mathbf{N}}(t)&=&Y_{i}^{\mathbf{N}}(0)+\underbrace{\sum_{q= 1}^{N_j}\sum_{\ell= 1}^{N_i} \int_{0}^{t}\sum_{j\sim i} {\bf 1}_{\left\{Y_{j}^{\mathbf{N}}(s-)=q, Y_{i}^{\mathbf{N}}(s-)=\ell\right\}}\left(\mathcal{N}^{\left(\ell,q\right)}_{\gamma_{ji} q\left(\frac{N_i-\ell}{N_i}\right)}(ds)\right)}_{{\scriptsize \textrm{Inter-transmission}}}\\
&&-\underbrace{\sum_{\ell= 1}^{N_i} \int_{0}^{t} {\bf 1}_{\left\{Y_i^{\mathbf{N}}(s-)=\ell\right\}}\left(\mathcal{N}^{\left(\ell\right)}_{\ell\mu_i}(ds)\right)}_{{\scriptsize \textrm{Healing}}}\nonumber,
\end{eqnarray}
where all Poisson point processes $\mathcal{N}^{\left(\ell,q\right)}_{\gamma_{ji} q\left(\frac{N_i-\ell}{N_i}\right)}$ and $\mathcal{N}^{\left(\ell\right)}_{\ell\mu_i}$ indexed by $\ell$ and $q$ are \textbf{independent}. This is an important fact from the peer-to-peer model that will be evoked latter in this Section and in Section~\ref{sec:extension}. Also, for notational simplicity, we drop from now on the upper-indexes of the Poisson processes. Observe that the inter-transmission term is the one that relies on the supertopology of the multipartite network. Next, we frame the normalized martingale term~$\left(\overline{\mathbf{M}}^\mathbf{N}(t)\right)$ hidden within the pathwise characterization of the normalized process~$\left(\mathbf{\overline{Y}}^\mathbf{N}(t)\right)$ and prove that it converges weakly to zero as the network $G^{\mathbf{N}}$ grows large. This loosely implies that the randomness of the process $\left(\mathbf{\overline{Y}}^\mathbf{N}(t)\right)$ dies out as the number of agents grows.
\subsection{Martingale Vanishes}
\label{sec:martingalevanishes}

We start by characterizing the martingale term~$\left(\mathbf{M}^\mathbf{N}(t)\right)$ of our macroprocess~$\left(\mathbf{Y}^\mathbf{N}(t)\right)$
\begin{equation}
\mathbf{Y}^\mathbf{N}(t)=\mathbf{Y}^\mathbf{N}(0)+\mathbf{M}^\mathbf{N}(t)+\int_0^t F\left(\mathbf{Y}^\mathbf{N}(s-)\right)ds\nonumber
\end{equation}
and afterwards we explore its structure to prove that its normalized counterpart
\begin{equation}
\left(\mathbf{\overline{M}}^\mathbf{N}(t)\right):=\left(M_1^\mathbf{N}(t)/N_1,M_2^\mathbf{N}(t)/N_2\right)\nonumber
\end{equation}
converges weakly to~$0$ as the number of nodes at each island goes to infinite with the number of islands~$M=2$ kept fixed. In words, this means that the randomness of the normalized macroprocess~$\left(\mathbf{\overline{Y}}^\mathbf{N}(t)\right)$ dies out as the network grows large. We start by showing that $\mathbf{\overline{M}}^\mathbf{N}(t)$ converges to zero in $\mathcal{L}_2$ for all $t$, $t\geq 0$. Then, by Doob's inequality, this will imply that it converges in probability, in the Skorokhod space of c\`{a}dl\`{a}g sample paths, to zero, as will be clearer momentarily. Finally, this implies that the martingale converges weakly to zero. As for the rest of Section~\ref{sec:meanfield}, we concentrate on the case of single virus spread over a bipartite network. The stochastic vector process~$\left(\mathbf{Y}^{\mathbf{N}}(t)\right)$ over the bipartite network admits the following pathwise characterization:
\begin{eqnarray}
Y_{i}^{\mathbf{N}}(t)&=&Y_{i}^{\mathbf{N}}(0)+\underbrace{\sum_{\ell=1}^{N_i}\sum_{q=1}^{N_j} \int_{0}^{t} {\bf 1}_{\left\{Y_{i}^{\mathbf{N}}(s-)=\ell, Y_{j}^{\mathbf{N}}(s-)=q\right\}}\left(\mathcal{N}^{\left(\ell,q\right)}_{\gamma_{ji} q\left(\frac{N_i-\ell}{N_i}\right)}(ds)\right)}_{{\scriptsize {\sf Inter-transmission}}}\label{eq:stochastic3}\\
&&-\underbrace{\sum_{\ell=1}^{N_i} \int_{0}^{t} {\bf 1}_{\left\{Y_i^{\mathbf{N}}(s-)=\ell\right\}}\left(\mathcal{N}^{\left(\ell\right)}_{\ell\mu_i}(ds)\right)}_{{\scriptsize {\sf Healing}}},\nonumber
\end{eqnarray}
for~$i,j\in\left\{1,2\right\}$ and~$i\neq j$. The \emph{inter-transmission} term in equation~(\ref{eq:stochastic3}) accounts for the number of infections transmitted from island $j$ to healthy nodes in island~$i$ up to time~$t$,~$t\geq 0$. The \emph{healing} term accounts for the number of healings that occur in island~$i$ during the time interval~$\left[0,t\right]$. One can check that almost surely the normalized process $\overline{Y}_i^{\mathbf{N}}(t)=Y_i^{\mathbf{N}}(t)/N_i\in\left[0,1\right],\,\,\,\forall{t\geq 0}$ and all $i=1,2$, if $\overline{Y}^N_i(0)\in\left[0,1\right]$ a.s., that is, the set $\left[0,1\right]\times\left[0,1\right]$ is invariant for the stochastic dynamics of $\left(\mathbf{\overline{Y}}^{\mathbf{N}}(t)\right)$ (refer to \cite{paper:CDC}). In words, the underlying stochastic dynamics are consistent with our intuition about the underlying meaning of $\left(\overline{Y}_i^\mathbf{N}(t)\right)$ that it is the fraction of infected nodes at island~$i$, and so it is clearly a quantity between $0$ and $1$, for all $t$, $t\geq 0$. The next Theorem states that equation~(\ref{eq:stochastic3}) can be further decomposed as equation~(\ref{eq:decomposed}) into a martingale $\left(\mathbf{\overline{M}}^{\mathbf{N}}(t)\right)$ plus a drift term, with an explicit characterization for the martingale term provided in the proof.
\begin{theorem}\label{th:dynkins}\textbf{(Process Decomposition)}
Let~$\left(Y_i^{\mathbf{N}}(t)\right)$ be the number of infected nodes at island~$i$ at time~$t$ for a bipartite network with two islands~$1$ and~$2$. Let~$\gamma_{ji}$ be the rate at which a node from island~$j$ attempts to infect a node in island~$i$. Then,~$\left(Y_i^{\mathbf{N}}(t)\right)$ admits the following pathwise characterization:
\begin{equation}
Y_i^{\mathbf{N}}(t)=Y_i^{\mathbf{N}}(0)+M_i^{\mathbf{N}}(t)+\int_{0}^{t}\gamma_{ji} Y_j^{\mathbf{N}}(s-)\left(\frac{N_i-Y_i^{\mathbf{N}}(s-)}{N_i}\right)ds-\int_{0}^{t}\mu_i Y_i^{\mathbf{N}}(s-)ds,\label{eq:decomposed}
\end{equation}
for~$i,j=1,2$, with~$i\neq j$, where~$\left(M_i^{\mathbf{N}}(t)\right)$ is a martingale and~$N_i$ is the number of nodes in island~$i$.
\end{theorem}
\begin{proof}
This follows as a Corollary to Dynkin's Lemma (see~\cite{Diffusion}), but we will provide an explicit characterization for the martingale $\left(M_1^{\mathbf{N}}(t),M_2^{\mathbf{N}}(t)\right)$ resulting from the microscopical diffusion model. It is easy to check that the well-known compensated Poisson process
\begin{equation}
M(t)=\mathcal{N}_{\gamma}(\left.\left(0,t\right.\right])-\gamma t\nonumber
\end{equation}
is a martingale adapted to the natural filtration~$\Sigma_{t}:=\sigma\left\{\mathcal{N}_\gamma(s)\,:\,0\leq s\leq t\right\}$ of~$\left(\mathcal{N}_\gamma(t)\right)$. Now, by compensating the Poisson point processes involved in equations~(\ref{eq:stochastic3}), we obtain
\begin{eqnarray}
Y_i^\mathbf{N}(t)&=&Y_i^\mathbf{N}(0)
+\underbrace{\sum_{\ell= 1}^{N_i}\sum_{q= 1}^{N_j} \int_{0}^{t} {\bf 1}_{\left\{Y_{j}^{\mathbf{N}}(s-)=q, Y_{i}^{\mathbf{N}}(s-)=\ell\right\}}\left(\mathcal{N}^{\left(\ell,q\right)}_{\gamma_{ji} q\left(\frac{N_i-\ell}{N_i}\right)}(ds)-\gamma_{ji} q\left(\frac{N_i-\ell}{N_i}\right)ds\right)}_{M^{\mathbf{N}}_{ai}(t)}
\nonumber
\\
\label{eq:stochastic2}
&&-\underbrace{\sum_{k= 1}^{N_i} \int_{0}^{t} {\bf 1}_{\left\{Y_i^{\mathbf{N}}(s-)=k\right\}}\left(\mathcal{N}_{k\mu_i}(ds)-k\mu_i ds\right)}_{M_{bi}^{\mathbf{N}}(t)}
\\
&&+\sum_{\ell= 1}^{N_i}\sum_{q= 1}^{N_j} \int_{0}^{t} {\bf 1}_{\left\{Y_{j}^{\mathbf{N}}(s-)=q, Y_{i}^{\mathbf{N}}(s-)=\ell\right\}}\gamma_{ji} q\left(\frac{N_i-\ell}{N_i}\right) ds\nonumber
-\sum_{k= 1}^{N_i} \int_{0}^{t} {\bf 1}_{\left\{Y_i^{\mathbf{N}}(s-)=k\right\}} k\mu_i ds.
\end{eqnarray}
It turns out that the martingale property of a compensated Poisson process is stable under the integration of a predictable process, \cite{Queues}, in particular, if $\left(Y(t)\right)$ is a c\`{a}dl\`{a}g adapted process, then
\begin{equation}
M(t)=\int_{0}^{t} Y(s-)\left(\mathcal{N}_{\lambda}(ds)-\lambda ds\right)\nonumber
\end{equation}
is a martingale adapted to the natural filtration~$\sigma\left\{\mathcal{N}_\lambda(s)\,:\,0\leq s\leq t\right\}$ (refer to \cite{Diffusion2}). Moreover, the space $\mathcal{M}$ of martingales adapted to the same filtration conforms a vector space and, therefore, $M_{ai}^{\mathbf{N}}(t)$ and $M_{bi}^{\mathbf{N}}(t)$ are martingales and
\begin{equation}
\label{eq:martingale}
M_i^{\mathbf{N}}(t)=M_{ai}^{\mathbf{N}}(t)-M_{bi}^{\mathbf{N}}(t)
\end{equation}
is a martingale. Therefore, the terms $M_{ai}^{\mathbf{N}}(t)$ and $M_{bi}^{\mathbf{N}}(t)$ in equation~(\ref{eq:stochastic2}) are martingales with respect to the natural filtration of the underlying Poisson point processes $\mathcal{N}^{\left(\ell,q\right)}_{\gamma_{ji} q\left(\frac{N_i-\ell}{N_i}\right)}$ and $\mathcal{N}^{\left(\ell\right)}_{\ell \mu_i}$, respectively, for $\ell\leq N_i$ and $q\leq N_j$.
Theorem~\ref{th:dynkins} is completed and the martingale term~$\left(\mathbf{M}^{\mathbf{N}}(t)\right)$ is characterized.
\end{proof}

Next, we prove that the variance of the normalized zero-mean martingales~$\left(\overline{M}_i^{\mathbf{N}}(t)\right)=\left(M_i^{\mathbf{N}}(t)/N_i\right)$ converge to zero (as $\mathbf{N}$ grows large) for all~$t$,~$t\geq 0$, and~$i=1,2$, that is $\overline{M}_i^{\mathbf{N}}(t)$ converges to zero in $\mathcal{L}_2$, for all time $t$.

\begin{theorem}
\label{th:martingalemorre}
Let $\left(\overline{M}_i^\mathbf{N}(t)\right)=\left(\overline{M}_{ai}^\mathbf{N}(t)-\overline{M}_{bi}^\mathbf{N}(t)\right)$ be described as in equation~(\ref{eq:stochastic2}). We have
\begin{equation}
{\sf Var}\left(\overline{M}_i^\mathbf{N}(t)\right)={\sf E}\left(\overline{M}_i^\mathbf{N}(t)\right)^2\longrightarrow 0,\,\,\,\forall{t\geq0}\mbox{ and $i=1,2$}\nonumber
\end{equation}
as $N_i\rightarrow\infty$ for $i=1,2$ with $N_i/N_j\rightarrow \alpha_{ji}<\infty$.
\end{theorem}

\begin{proof}
Note first that for $s<t$
\begin{equation}
{\sf E}\left(M_i^{\mathbf{N}}(t)\right)={\sf E}\left({\sf E}\left(M_i^{\mathbf{N}}(t)|\mathcal{F}_s\right)\right)={\sf E}\left(M_i^{\mathbf{N}}(s)\right).\nonumber
\end{equation}
So, $\left(M_i^{\mathbf{N}}(t)\right)$ is zero-mean, ${\sf E}\left(M_i^{\mathbf{N}}(t)\right)={\sf E}\left(M_i^{\mathbf{N}}(0)\right)=0$ for all $t\geq 0$; thus, ${\sf Var}\left(\overline{M}_i^{\mathbf{N}}(t)\right)={\sf E}\left(\overline{M}_i^{\mathbf{N}}(t)\right)^2$. Before proceeding, we state that orthogonality of compensated Poisson martingales is preserved under integration of predictable processes, which will be important to complete the proof of Theorem~\ref{th:martingalemorre}.
\begin{theorem}[Orthogonality under Integration]\label{th:ortho}
Let: $\left(\mathcal{N}_{\gamma_1}(t)\right)$ and $\left(\mathcal{N}_{\gamma_2}(t)\right)$ be two independent Poisson processes; $M_1(t)=\mathcal{N}_{\gamma_1}(t)-\gamma_1 t$ and $M_2(t)=\mathcal{N}_{\gamma_2}(t)-\gamma_2 t$ be the corresponding compensated martingales; $\left(F_1(t)\right)$ and $\left(F_2(t)\right)$ be almost surely bounded predictable processes with respect to the natural filtrations $\left(\sigma\left(M_1(t);t\leq s\right)\right)_{s\geq 0}$ and $\left(\sigma\left(M_2(t);t\leq s\right)\right)_{s\geq 0}$, respectively. Then,
\begin{equation}
\underbrace{\left(\int_0^t F_1(s)dM_1(s)\right)}_{M_a(t)}\underbrace{\left(\int_0^t F_2(s)dM_2(s)\right)}_{M_b(t)}\mbox{ is a martingale.}\nonumber
\end{equation}
\end{theorem}

\begin{proof}
Refer to \cite{Queues}.
\end{proof}

In particular, it follows as a Corollary to Theorem~\ref{th:ortho} that
\begin{eqnarray}
\label{eq:ortogo1}
E\left[\left(\int_0^t F_1(s)dM_1(s)\right)\left(\int_0^t F_2(s)dM_2(s)\right)\right] &\!\!\!\! =\!\!\!\! & E\left[\left(\int_0^0 F_1(s)dM_1(s)\right)
\left(\int_0^0 F_2(s)dM_2(s)\right)\right]
\\
& \!\!\!\! =\!\!\!\! & 0,\,\,\forall{t\geq 0}.\label{eq:ortogo2}
\end{eqnarray}
In words, the resulting integral martingales $\left(M_{ai}(t)\right)$ and $\left(M_{bi}(t)\right)$ are orthogonal.
Now, back to the proof of Theorem~\ref{th:martingalemorre}.
\begin{eqnarray}
{\sf E}\left(M_{ai}^{\mathbf{N}}(t)\right)^2&=& {\sf E}\left( \sum_{\ell= 1}^{N_i}\sum_{q= 1}^{N_j}\int_{0}^{t} {\bf 1}_{\left\{Y_{j}^{\mathbf{N}}(s-)=q, Y_{i}^{\mathbf{N}}(s-)=\ell\right\}}\left(\mathcal{N}_{\gamma_{ji} q\left(\frac{N_i-\ell}{N_i}\right)}(ds)-\gamma_{ji}\right. q\left(\frac{N_i-\ell}{N_i}\right)ds\right)^2\nonumber\\
&=&\sum_{\ell= 1}^{N_i}\sum_{q= 1}^{N_j}{\sf E}\left(\int_{0}^{t} {\bf 1}_{\left\{Y_{j}^{\mathbf{N}}(s-)=q, Y_{i}^{\mathbf{N}}(s-)=\ell\right\}}\left(\mathcal{N}_{\gamma_{ji} q\left(\frac{N_i-\ell}{N_i}\right)}(ds)-\gamma_{ji} q\left(\frac{N_i-\ell}{N_i}\right)ds\right)\right)^2\nonumber\\
&=&\sum_{\ell= 1}^{N_i}\sum_{q= 1}^{N_j}{\sf E}\left(\int_{0}^{t}{\bf 1}_{\left\{Y_{j}^{\mathbf{N}}(s-)=q, Y_{i}^{\mathbf{N}}(s-)=\ell\right\}}\gamma_{ji} q\left(\frac{N_i-\ell}{N_i}\right)ds\right)\nonumber\\
&\leq&{\sf E}\left(\int_{0}^{t}\sum_{\ell= 1}^{N_i}\sum_{q= 1}^{N_j}{\bf 1}_{\left\{Y_{j}^{\mathbf{N}}(s-)=q, Y_{i}^{\mathbf{N}}(s-)=\ell\right\}}\gamma_{ji} N_jds\right)\nonumber\\
&\leq&\gamma_{ji}N_jt,\nonumber
\end{eqnarray}
where the second equality follows from Theorem~\ref{th:ortho}, remarking that all Poisson processes involved, indexed by $\ell$ and $q$, are independent and thus the integral martingales are orthogonal and the cross-terms cancel out. The third equality is due to the It\^{o} isometry Theorem (refer to~\cite{Diffusion2} or~\cite{Karatzas}) and the fact that the quadratic variation of a compensated Poisson martingale is given by $\left\langle\mathcal{N}_{\gamma}(t)-\gamma t \right\rangle=\gamma t$. The last inequality holds since the family of sets $I_{\ell,q}:=\left\{s\in\mathbb{R}\,:\,Y_{j}^{\mathbf{N}}(s-)=q, Y_{i}^{\mathbf{N}}(s-)=\ell\right\}$, indexed by $\ell,q$, are disjoint and thus
\begin{equation}
\sum_{\ell= 1}^{N_i}\sum_{q= 1}^{N_j}{\bf 1}_{\left\{Y_{j}^{\mathbf{N}}(s-)=q, Y_{i}^{\mathbf{N}}(s-)=\ell\right\}}=
{\bf 1}_{\bigcup_{\ell,q}{\left\{Y_{j}^{\mathbf{N}}(s-)=q, Y_{i}^{\mathbf{N}}(s-)=\ell\right\}}}\leq {\bf 1}_{\left[0,1\right]}(s).\nonumber
\end{equation}
Therefore,
\begin{equation}
{\sf E}\left(\overline{M}_{ai}^{\mathbf{N}}(t)\right)^2=\frac{1}{N_i^2}{\sf E}\left(M_{ai}^{\mathbf{N}}(t)\right)^2\leq \frac{1}{N_i}\frac{\gamma_{ji}}{4}\left(\frac{N_j}{N_i}\right)t.\nonumber
\end{equation}
Thus,
\begin{equation}
{\sf E}\left(\overline{M}_{ai}^{\mathbf{N}}(t)\right)^2\longrightarrow 0\nonumber
\end{equation}
as $N_i\rightarrow \infty$ and $N_j\rightarrow \infty$, and $\frac{N_j}{N_i}\rightarrow \alpha_{ij}<\infty$. Similarly, the variance of the martingale $\left(\overline{M}_{bi}^{\mathbf{N}}(t)\right)$ converges to zero. Thus, the martingale vanishes in $\mathcal{L}_2$ with $O\left(1/N_i\right)$.
\end{proof}

Since $\left(\overline{M}_i^{\mathbf{N}}(t)\right)$ is a martingale, from Doob's inequality:
\begin{equation}
P\left(\sup_{0\leq t\leq T}\left|\overline{M}_i^{\mathbf{N}}(t)\right|>\epsilon\right)\leq \frac{{\sf E}\left(\overline{M}_i^{\mathbf{N}}(T)\right)^2}{\epsilon^2}\longrightarrow 0,\,\,\,\forall\epsilon>0,\,\,\,\forall T\geq0.\label{eq:Doobsinequality}
\end{equation}
That is, $\left(\overline{M}_i^\mathbf{N}(t)\right)$ converges to zero in probability in the space of \emph{c\`{a}dl\`{a}g} paths with the \emph{sup} norm.

The next Theorem is an extension for stochastic processes of the statement that convergence in probability implies convergence in distribution for real valued random variables. The Theorem is is Proposition $C.5$ from~\cite{Queues},  and it will imply that $\left(\mathbf{\overline{M}}^{\mathbf{N}}(t)\right)$ converges weakly (in the Skorokhod topology) to $0$. That is, $P_{\overline{\mathbf{M}}^{\mathbf{N}}}$, the probability measure over $D_{\mathbb{R}}\left[0,T\right]$ induced by $\left(\mathbf{\overline{M}}^{\mathbf{N}}(t)\right)$, converges weakly to $\delta_0$ (Dirac measure about $m(t)=0$).
\begin{theorem}
\label{th:probimplyweak}
Let $\left(Z^N(t)\right)$ be a sequence of c\`{a}dl\`{a}g stochastic processes on the interval $\left[0,T\right]$ such that,
\begin{equation}
\label{eq:doobs2}
P\left(\sup_{0\leq t\leq T}\left|Z^{N}(t)-z(t)\right|>\epsilon\right)\longrightarrow 0,\,\,\,\forall\epsilon>0,\nonumber
\end{equation}
where $\left(z(t)\right)$ is a deterministic c\`{a}dl\`{a}g function over $\left[0,T\right]$. Then,
\begin{equation}
\left(Z^N(t)\right)\Rightarrow \left(z(t)\right)\mbox{ on $\left[0,T\right]$},\nonumber
\end{equation}
i.e., $\left(Z^N(t)\right)$ converges weakly, for the Skorokhod topology, to $\left(z(t)\right)$. In other words, the sequence of probability measures $P_{Z^{N}}$ over $D_{\mathbb{R}}\left[0,T\right]$ induced by $\left(Z^N(t)\right)$ converges weakly to $\delta_{\left(z(t)\right)}$ (Dirac measure about $\left(z(t)\right)$).
\end{theorem}

\begin{proof}
Refer to proposition $C.5$ in \cite{Queues}.
\end{proof}

Thus, we conclude that the martingale $\left(\overline{M}_i^{\mathbf{N}}(t)\right)$ converges weakly to zero.

\subsection{$\left(\mathbf{\overline{Y}}^{\mathbf{N}}(t)\right)$ converges weakly }
\label{sec:tight}

The next Theorem is a stochastic version of Arzel\`{a}-Ascoli and will provide sufficient conditions to guarantee the tightness of the sequence $\left(\mathbf{\overline{Y}}^{\mathbf{N}}(t)\right)$--just as Arzel\`{a}-Ascoli provides sufficient conditions to guarantee tightness of a family of (deterministic) functions.
\begin{theorem}\label{th:arzela}
Let $\left(\mathbf{\overline{Y}}^{\mathbf{N}}(t)\right)$ be a sequence of c\`{a}dl\`{a}g processes. Then, the sequence of probability measures $P^{\mathbf{N}}$ induced on $D_{\mathbb{R}}\left[0,T\right]$ by $\left(\mathbf{\overline{Y}}^{\mathbf{N}}(t)\right)$ is tight and any weak limit point of this sequence is concentrated on the subset of continuous functions $C_{\mathbb{R}}\subset D_{\mathbb{R}}$ if and only if the following two conditions hold for each $T>0$ and $\epsilon>0$:
\begin{eqnarray}
\lim_{k\rightarrow\infty}\limsup_{N\rightarrow\infty}P\left(\sup_{0\leq t\leq T} \overline{Y}^{\mathbf{N}}(t)\geq k\right) &= &0\mbox{         (Uniform Boundness)}\label{eq:arzela}\\
\lim_{\delta\rightarrow 0}\limsup_{N\rightarrow\infty}P\left(\omega(\overline{Y}^{\mathbf{N}},\delta,T)\geq \epsilon\right)&= &0\mbox{         (Equicontinuity)}\label{eq:arzela2}
\end{eqnarray}
where we defined
\begin{equation}
\omega(x, \delta, T)=\sup\left\{\sup_{u,v\in\left[t,t+\delta\right]}|x(u)-x(v)|\,:\,0\leq t\leq t+\delta\leq T\right\}.\nonumber
\end{equation}
\end{theorem}
\begin{proof}
Refer to \cite{Ruth}.
\end{proof}

From Theorems~\ref{th:martingalemorre} and~\ref{th:probimplyweak}, we have that $\left(\overline{\mathbf{M}}^{\mathbf{N}}(t)\right)\Rightarrow 0$. We are just left to show that our sequence $\left(\overline{Y}_i^{\mathbf{N}}(t)\right)$ meets the requirements in equations~(\ref{eq:arzela}) and~(\ref{eq:arzela2}) and therefore it is tight (or relatively compact). In other words, it admits a convergent subsequence $\left(\overline{Y}^{\mathbf{N}_p}(t)\right)\rightarrow \overline{Y}(t)$ where $\left(\overline{Y}(t)\right)$ is almost surely continuous. Indeed, $P\left(\sup_{0\leq t\leq T}\overline{Y}_i^{\mathbf{N}}(t)\geq k\right)=0$, $\forall k>1$, and the first condition holds trivially. The second condition is a stochastic version of the equicontinuity condition in the Arzel\`{a}-Ascoli Theorem.
\begin{eqnarray}
\label{eq:variation}
\omega\left(\overline{Y}_i^{\mathbf{N}}, \delta, T\right)&=&\sup\left\{\sup_{u,v\in\left[t,t+\delta\right]}\left|\overline{Y}_i^{\mathbf{N}}(u)-\overline{Y}_i^{\mathbf{N}}(v)\right|\,:\,0\leq t\leq t+\delta\leq T\right\}\label{eq:variation1}\\
&=& \sup_{0\leq t\leq t+\delta\leq T}\left\{\sup_{u,v\in\left[t,t+\delta\right]}\left|\overline{M}_i^{\mathbf{N}}(u)-\overline{M}_i^{\mathbf{N}}(v)\right.\right.\label{eq:variation2}\\
&&\left.\left.+\int_{u}^{v}\gamma_{ji}\frac{N_j}{N_i} \overline{Y}_j^{\mathbf{N}}(s-)\left(1-\overline{Y}_i^{\mathbf{N}}(s-)\right)ds-\int_{u}^{v}\mu_i \overline{Y}_i^{\mathbf{N}}(s-)ds\right|\right\}\nonumber\\
&\leq& \sup_{0\leq t\leq t+\delta\leq T}\left\{\sup_{u,v\in\left[t,t+\delta\right]}\left|\overline{M}_i^{\mathbf{N}}(u)-\overline{M}_i^{\mathbf{N}}(v)\right|\right\}+\gamma_{ji}\frac{N_j}{N_i}\frac{\delta}{4}\label{eq:variation4}\\
&:=&\omega_2\left(\overline{Y}_i^{\mathbf{N}},\delta,T\right).\label{eq:variation5}
\end{eqnarray}
Although in equation~\eqref{eq:variation4} there is an explicit~$j$ dependence, we note that, for a bipartite network, choosing island~$i$ fixes the other island~$j$; because of this, the arguments in $\omega_2\left(\overline{Y}_i^{\mathbf{N}},\delta,T\right)$ in equation~\eqref{eq:variation5} do not show explicitly the dependence in~$j$. For any $\epsilon>0$, we have
\begin{equation}
P\left(\omega\left(\overline{Y}_i^{\mathbf{N}},\delta,T\right)\geq \epsilon\right)\leq P\left(\omega_2\left(\overline{Y}_i^{\mathbf{N}},\delta,T\right)\geq\epsilon\right).\nonumber
\end{equation}
Now, from equation~(\ref{eq:Doobsinequality}) and any $\alpha>0$, we can choose $N_1$ and $N_2$ large enough so that
\begin{equation}
P\left(\sup_{0\leq t\leq T}\left|\overline{M}_i^{\mathbf{N}}(t)\right|>\epsilon\right)< \alpha.\nonumber
\end{equation}
and therefore,
\begin{equation}
P\left(\omega\left(\overline{Y}_i^{\mathbf{N}},\delta,T\right)\geq \epsilon\right)\leq P\left(\omega_2\left(\overline{Y}_i^{\mathbf{N}},\delta,T\right)\geq\epsilon\right)<\alpha\nonumber
\end{equation}
for $\delta$ small enough. We conclude that $\left(\mathbf{\overline{Y}}^{\mathbf{N}}(t)\right)$ is a tight family. In other words, it admits a convergent subsequence $\left(\overline{\mathbf{Y}}^{\mathbf{N}_p}(t)\right)\Rightarrow \left(\overline{\mathbf{Y}}(t)\right)$ where $\left(\overline{\mathbf{Y}}(t)\right)$ is almost surely continuous. In fact, the sequence $\left(\overline{\mathbf{Y}}^{\mathbf{N}}(t)\right)$ is not only tight, but it converges to the solution of a deterministic differential equation as shown in the next Theorem~\ref{eq:finalconvergence}. The main argument is that any weak accumulation point in the tight sequence $\left(\overline{\mathbf{Y}}^{\mathbf{N}}(t)\right)$ should obey the equation
\begin{equation}
\frac{d}{dt}\mathbf{\overline{Y}}(t)=\mathbf{F}\left(\mathbf{\overline{Y}}(t)\right)\label{eq:diferenciale}
\end{equation}
and from the uniqueness of equation~(\ref{eq:diferenciale}), the whole sequence converges.
\begin{theorem}\label{eq:finalconvergence}
Let $\left(\mathbf{\overline{Y}}^{\mathbf{N}_p}(t)\right)$ be a subsequence converging weakly to $\left(\mathbf{\overline{Y}}(t)\right)$ (an almost surely continuous process) with $\mathbf{\overline{Y}}^{\mathbf{N}_p}(0)\Rightarrow \mathbf{\overline{Y}}(0)$ and let $\frac{N_i}{N_j}\rightarrow \alpha_{ij}\in\mathbb{R}^{+}$. Then
\begin{equation}
\overline{Y}_i(t)=\overline{Y}_i(0)+\int_{0}^{t}\overline{\gamma}_{ji} \overline{Y}_j(s)\left(1-\overline{Y}_i(s)\right)ds-\int_{0}^{t}\mu_i \overline{Y}_i(s)ds,\label{eq:integr}
\end{equation}
where we defined $\overline{\gamma}_{ji}:=\gamma_{ji}\alpha_{ji}$.
\end{theorem}
\begin{proof}
\hspace{0.8cm}
We have the following term by term convergence
\begin{equation}
\begin{array}{cccccccccc}
\overline{Y}_i^{\mathbf{N}_p}(t)&=&\overline{Y}_i^{\mathbf{N}_p}(0)&+& \overline{M}^{\mathbf{N}_p}_i(t)&+&\int_{0}^{t}\gamma_{ji}\frac{N_j}{N_i}\overline{Y}_j^{\mathbf{N}_p}(s)\left(1-\overline{Y}_i^{\mathbf{N}_p}(s)\right)ds&-&\int_{0}^{t}\mu_i\overline{Y}_i^{\mathbf{N}_p}(s)ds\\
\downarrow 1 & & \downarrow 2 & & \downarrow 3 & & \downarrow 4 & & \downarrow 5\\
\overline{Y}_i(t)&=&\overline{Y}_i(0)&+& 0 &+& \int_{0}^{t}\overline{\gamma}_{ji}\overline{Y}_j(s)\left(1-\overline{Y}_i(s)\right)ds&-&\int_{0}^{t}\mu_i\overline{Y}_i(s)ds
\end{array}.\nonumber
\end{equation}
Convergences $1$ and $2$ hold since we assumed that $\left(\overline{Y}_i^{\mathbf{N}_p}(t)\right)\Rightarrow\left(\overline{Y}_i(t)\right)$ and $\overline{Y}_i^{\mathbf{N}_p}(0)\Rightarrow \overline{Y}_i(0)$. Convergence $3$ holds since $\mathbf{\overline{M}}^{\mathbf{N}}(t)\Rightarrow 0$ as proved before. Moreover, since $\left(\mathbf{\overline{Y}}(t)\right)$ is almost surely continuous, then $\left(\overline{Y}_i^{\mathbf{N}_p}(t)\right)\rightarrow\left(\overline{Y}_i(t)\right)$ almost surely uniformly over compact intervals. Therefore, given $T>0$,
\begin{equation}
\int_{0}^{t}\gamma_{ji}\frac{N_j}{N_i}\overline{Y}_j^{\mathbf{N}_p}(s)\left(1-\overline{Y}_i^{\mathbf{N}_p}(s)\right)ds\rightarrow
\int_{0}^{t}\overline{\gamma}_{ji}\overline{Y}_j(s)\left(1-\overline{Y}_i(s)\right)ds,\nonumber
\end{equation}
almost surely uniformly over the interval $\left[0,T\right]$ and convergence $4$ follows. The $5$th case results similarly. Therefore, for any convergent subsequence with $\left(\mathbf{\overline{Y}}^{\mathbf{N}_p}(t)\right)\Rightarrow\left(\mathbf{\overline{Y}}(t)\right)$ and $\mathbf{\overline{Y}}^{\mathbf{N}_p}(0)\Rightarrow\mathbf{\overline{Y}}(0)$, it follows that $\left(\mathbf{\overline{Y}}(t)\right)$ is solution of the integral equation~(\ref{eq:integr}).
\end{proof}

Finally, the next theorem rigorously states the emergent dynamics of the single virus spread over bipartite networks as the number of agents grows large.

\begin{theorem}\label{th:convbipartite}
Let $\left(\overline{Y}^{\mathbf{N}}_1(0), \overline{Y}^{\mathbf{N}}_2(0)\right)\Rightarrow \mathbf{y}_0\in \mathbb{R}^2$. Then, the normalized sequence $\left(\overline{\mathbf{Y}}^{\mathbf{N}}(t)\right)$ converges weakly to the solution $\left(y_1\left(t,\mathbf{y}_0\right),y_1\left(t,\mathbf{y}_0\right)\right)$ of the following ODE:
\begin{eqnarray}
\frac{d}{dt}y_1(t) & = & \left(\overline{\gamma}_{21}y_2(t)\right)\left(1-y_1(t)\right)-\mu_1 y_1(t)\label{eq:biv1}\\
\frac{d}{dt}y_2(t) & = & \left(\overline{\gamma}_{12}y_1(t)\right)\left(1-y_2(t)\right)-\mu_2 y_2(t)\label{eq:biv2}
\end{eqnarray}
\end{theorem}

\begin{proof}
As the underlying vector field of~(\ref{eq:integr}), $\mathbf{F}=\left(F_1,F_2\right)\,:\,\left[0,1\right]^2\rightarrow \mathbb{R}^2$, where $F_i\left(y_1,y_2\right):=\overline{\gamma}_{ji}y_j\left(1-y_i\right)-\mu_i y_i$, is Lipschitz, the continuous (and thus, differentiable) solution $\left(\mathbf{\overline{Y}}(t)\right)$ of (\ref{eq:integr}) is unique. Thus, any weak limit of $\left(\mathbf{\overline{Y}}^{\mathbf{N}}(t)\right)$ with initial condition given by $\left(\mathbf{\overline{Y}}^{\mathbf{N}}(0)\right)$ and converging weakly to $\mathbf{\overline{Y}}(0)$ is equal to the unique solution $\left(\mathbf{\overline{Y}}(t)\right)$ of~(\ref{eq:integr}) with initial condition $\left(\mathbf{\overline{Y}}(0)\right)$. Therefore, the whole sequence converges $\left(\mathbf{\overline{Y}}^{\mathbf{N}}(t)\right)\Rightarrow\left(\mathbf{\overline{Y}}(t)\right)$ to the solution of~(\ref{eq:integr}). Equation~(\ref{eq:integr}) is the integral version the ODE~(\ref{eq:biv1})-(\ref{eq:biv2}). Theorem~(\ref{th:convbipartite}) is concluded.
\end{proof}

We showed in this Section that the sequence $\left(\mathbf{\overline{Y}^{N}}(t)\right)$ over the corresponding sequence of bipartite networks $G^{\left(N_1,N_2\right)}$ converges weakly to the solution of a vector differential equation. We explored the martingale structure of the perturbing noise on the dynamics of the process $\left(\mathbf{\overline{Y}^{N}}(t)\right)$ to show that it converges weakly under the Skorokhod topology to zero. As a Corollary to this fact, the family $\left(\mathbf{\overline{Y}^{N}}(t)\right)$ is tight with a single accumulation point given by the (unique) solution of a limiting differential equation. Since any convergent subsequence converges to the same accumulation point, then the whole sequence $\left(\mathbf{\overline{Y}^{N}}(t)\right)$ converges to the unique accumulation point, namely, the solution of the ODE~(\ref{eq:biv1})-(\ref{eq:biv2}). In the next section, we extend the convergence result to the multi-virus multipartite network case.

\section{Mean Field -- Multivirus over Multipartite Networks}\label{sec:extension}

In Section~\ref{sec:meanfield}, we established that the sequence of single virus macroprocesses $\left(\overline{Y}_1^{\mathbf{N}}(t),\overline{Y}_2^{\mathbf{N}}(t)\right)$ over the corresponding sequence of bipartite networks $G^{\mathbf{N}}$ converges weakly to the solution $\left(y_1(t),y_2(t)\right)$ of a deterministic ODE given by equations~(\ref{eq:biv1})-(\ref{eq:biv2}). We divided the proof into four main steps:
\begin{enumerate}[(i)]
 \item The martingale $\left(\overline{M}_1^{\mathbf{N}}(t),\overline{M}_2^{\mathbf{N}}(t)\right)$ converges weakly to zero;
 \item The family $\left(\overline{Y}_1^{\mathbf{N}}(t),\overline{Y}_2^{\mathbf{N}}(t)\right)$, indexed by $\mathbf{N}$, is tight;
 \item Any accumulation point of the tight family is solution of~(\ref{eq:biv1})-(\ref{eq:biv2});
 \item Uniqueness of the differential equation~(\ref{eq:biv1})-(\ref{eq:biv2}) implies convergence.
 \end{enumerate}
 In this Section, we extend Theorem~\ref{th:convbipartite} to the more general case of multi-virus epidemics over multipartite networks. We consider in Subsection~\ref{subsec:singlevirus} single-virus over multipartite networks and then the general case of multivirus over multipartite networks is in Subsection~\ref{subsec:multivirus}.

\subsection{Single-virus over Multipartite Networks}
\label{subsec:singlevirus}

For single virus spread, remark that $\left(\mathbf{\overline{Y}}^{\mathbf{N}}(t)\right)=\left(\overline{Y}_{1}^{\mathbf{N}}(t),\ldots,\overline{Y}_{M}^{\mathbf{N}}(t)\right)$ stands for the process associated with the fraction of infected nodes at each island $i \in \left\{1,\ldots,M\right\}$ over the multipartite network $G^{\mathbf{N}}$ with $M$ islands and supertopology induced by the topology of the graph $G$. In Section~\ref{sec:meanfield}, we obtain the following pathwise description for $\left(\mathbf{\overline{Y}}^{\mathbf{N}}(t)\right)$
\begin{equation}
Y_{i}^{\mathbf{N}}(t) =  Y_{i}^{\mathbf{N}}(0)+M_{i}^{\mathbf{N}}(t)+\sum_{j\sim i}\int_{0}^{t}\gamma_{ji} Y_{j}^{\mathbf{N}}(s-)\left(\frac{N_i-Y_{i}^{\mathbf{N}}(s-)}{N_i}\right)ds-\int_{0}^{t}\mu_i Y_{i}^{\mathbf{N}}(s-)ds,\nonumber
\end{equation}
for $i=1,\ldots,M$. The infections from all neighboring islands are now coupled by these $M$ equations. The corresponding martingales are given by
\begin{eqnarray}
&&M_{i}^{\mathbf{N}}(t)  =
\nonumber
\\
&& \sum_{j\sim i}\sum_{\ell=1}^{N_i}\sum_{q = 1}^{N_j}
\underbrace{\int_{0}^{t}{\bf 1}_{\left\{Y_{j}^{\mathbf{N}}(s-)=q, Y_{i}^{\mathbf{N}}(s-)=\ell\right\}}
\times\left(\mathcal{N}_{\gamma_{ji} q\left(\frac{N_i-\ell}{N_i}\right)}(ds)-\gamma_{ji} q\left(\frac{N_i-\ell}{N_i}\right)ds\right)}_{:=M^{\mathbf{N}}_{aij}(t,\ell,q)}\nonumber\\
&&-\underbrace{\sum_{\ell= 1}^{N_i} \int_{0}^{t} {\bf 1}_{\left\{Y_{ik}^{\mathbf{N}}(s-)=\ell\right\}}\left(\mathcal{N}_{\ell\mu_i}(ds)-\ell\mu_i ds\right)}_{:=M_b^{\mathbf{N}}(t)}\nonumber
\end{eqnarray}
Next, we prove that the sequence $\left(\mathbf{\overline{Y}}^{\mathbf{N}}(t)\right)$ over the underlying sequence of multipartite networks $G^{\mathbf{N}}$ converges weakly to the solution of an ODE.
\begin{theorem}\label{th:convmulti}
Let $\left(\mathbf{\overline{Y}}^{\mathbf{N}}(0)\right)\Rightarrow \mathbf{y}_0\in \mathbb{R}^M$ with $\mathbf{N}=\left(N_1,\ldots,N_M\right)\rightarrow \infty$ and $\frac{N_i}{N_j}\rightarrow \alpha_{ij}<\infty$ for all $i\sim j$. Then, the normalized sequence $\left(\overline{\mathbf{Y}}^{\mathbf{N}}(t)\right)$ converges weakly to the solution $\left(\mathbf{y}(t,\mathbf{y}_0)\right)=\left(y_1\left(t,\mathbf{y}_0\right),\ldots, y_M\left(t,\mathbf{y}_0\right)\right)$ of the following ODE:
\begin{equation}
\frac{d}{dt}y_i(t) =  \underbrace{\left(\sum_{j\sim i}\overline{\gamma}_{ji}y_j(t)\right)\left(1-y_i(t)\right)-\mu_i y_i(t)}_{:=F_i\left(y_1(t),\ldots,y_M(t)\right)}\label{eq:mult1}
\end{equation}
with $\overline{\gamma}_{ji}:=\alpha_{ji}\gamma_{ji}$.
\end{theorem}

\begin{proof}
For the sake of clarity, we revisit each of the points (i)-(iv) referred to in the beginning of this Section.

\textbf{$\left(\mbox{i}\right)$ Martingale vanishes.} We start by observing that, for fixed $i,j\in \left\{1,\ldots,M\right\}$, the underlying Poisson point processes $\mathcal{N}_{\gamma_{ji} q\left(\frac{N_i-\ell}{N_i}\right)}$ indexed by $\left(\ell,q\right)\in\left\{1,\ldots,N_i\right\}\times\left\{1,\ldots,N_j\right\}$ are independent as explained in Section~\ref{sec:pointprocesses}. Now, fixing only $i$, the Poisson processes are still independent due to the independence of the exponential time services associated with different nodes. Therefore, it follows as a Corollary to Theorem~\ref{th:ortho} that, for fixed $i$, the martingales $M^{\mathbf{N}}_{aij}(t,\ell,q)$, indexed by $j$, $\ell,q$ are orthogonal in the same sense as in equations~(\ref{eq:ortogo1})-(\ref{eq:ortogo2}). Let $M^{\mathbf{N}}_{ai}(t):=\sum_{j\sim i}\sum_{\ell=1}^{N_i}\sum_{q = 1}^{N_j}M^{\mathbf{N}}_{aij}(t,\ell,q)$. It turns out that
\begin{eqnarray*}
{\sf E}\left(M_{ai}^{\mathbf{N}}(t)\right)^2&\!\!\!=\!\!\!& {\sf E}\left( \sum_{j\sim i}\sum_{\ell= 1}^{N_i}\sum_{q= 1}^{N_j} M^{\mathbf{N}}_{aij}(t,\ell,q)\right)^2
=\sum_{j\sim i}\sum_{\ell= 1}^{N_i}\sum_{q= 1}^{N_j}{\sf E}\left(M^{\mathbf{N}}_{aij}(t,\ell,q)\right)^2\nonumber\\
&\!\!\!=\!\!\!&\sum_{j\sim i}\sum_{\ell= 1}^{N_i}\sum_{q= 1}^{N_j}{\sf E}\left(\int_{0}^{t}{\bf 1}_{\left\{Y_{j}^{\mathbf{N}}(s-)=q, Y_{i}^{\mathbf{N}}(s-)=\ell\right\}}\gamma_{ji} q\left(\frac{N_i-\ell}{N_i}\right)ds\right)\nonumber\\
&\!\!\!\leq\!\!\!&\sum_{j\sim i}{\sf E}\left(\int_{0}^{t}\sum_{\ell= 1}^{N_i}\sum_{q= 1}^{N_j}{\bf 1}_{\left\{Y_{j}^{\mathbf{N}}(s-)=q, Y_{i}^{\mathbf{N}}(s-)=\ell\right\}}\gamma_{ji} N_jds\right)\nonumber\\
&\!\!\!\leq\!\!\!&\sum_{j\sim i}\gamma_{ji}N_jt \leq M \max_{j=1,\ldots,M}\left\{\gamma_{ji}N_j\right\}t,\nonumber
\end{eqnarray*}
and thus,
\begin{equation}
{\sf E}\left(\overline{M}_{ai}^{\mathbf{N}}(t)\right)^2=\frac{1}{N_i^2}{\sf E}\left(M_{ai}^{\mathbf{N}}(t)\right)^2\leq \frac{M}{N_i}\left(\max_{j=1,\ldots,M}\left\{\frac{\gamma_{ji}N_j}{N_i}\right\}\right)t\rightarrow 0.\nonumber
\end{equation}
Similarly ${\sf E}\left(\overline{M}_{bi}^{\mathbf{N}}(t)\right)^2\rightarrow 0$ and therefore, the normalized martingale $\overline{M}_{i}^{\mathbf{N}}(t)$ converges to zero in $\mathcal{L}_2$ for all time $t$, $t\geq 0$. Now, from Doob's inequality and from Theorem~\ref{th:probimplyweak}, we conclude that $\left(\overline{M}_{i}^{\mathbf{N}}(t)\right)\Rightarrow 0$.

\textbf{$\left(\mbox{ii}\right)$ The family $\left(\mathbf{\overline{Y}}^{\mathbf{N}}(t)\right)$ is tight.}
As in equations~(\ref{eq:variation1})-(\ref{eq:variation5}), we have
\begin{eqnarray}
\label{eq:multivariation}
\omega\left(\overline{Y}_i^{\mathbf{N}}, \delta, T\right)&\leq& \sup_{0\leq t\leq t+\delta\leq T}\left\{\sup_{u,v\in\left[t,t+\delta\right]}\left|\overline{M}_i^{\mathbf{N}}(u)-\overline{M}_i^{\mathbf{N}}(v)\right|\right\}
+\sum_{j\sim i}\gamma_{ji}M\frac{N_j}{N_i}\delta\\
&:=&\omega_2\left(\overline{Y}_i^{\mathbf{N}},\delta,T\right).
\end{eqnarray}
From Theorem~\ref{th:arzela} and similar arguments as in Section~\ref{sec:meanfield}, we conclude that $\left(\mathbf{\overline{Y}}^{\mathbf{N}}(t)\right)$ is a tight family, that is, it admits a convergent subsequence $\left(\mathbf{\overline{Y}}^{\mathbf{N_k}}(t)\right)\Rightarrow \left(\mathbf{\overline{Y}}(t)\right)$. Also, from Theorem~\ref{th:arzela}, $\left(\mathbf{\overline{Y}}(t)\right)$ is almost surely continuous.

\textbf{$\left(\mbox{iii}\right)$ If $\left(\mathbf{\overline{Y}}^{\mathbf{N}_k}(t)\right)\Rightarrow \left(\mathbf{\overline{Y}}(t)\right)$ then, $\left(\mathbf{\overline{Y}}(t)\right)$ is solution of the ODE~(\ref{eq:mult1}).} It follows similarly to as done in the proof of Theorem~(\ref{eq:finalconvergence}), remarking that we assume a finite (fixed) number of islands $M$.

\textbf{$\left(\mbox{iv}\right)$ $\left(\mathbf{\overline{Y}}^{\mathbf{N}}(t)\right)\Rightarrow \left(\mathbf{\overline{Y}}(t)\right)$, where $\left(\mathbf{\overline{Y}}(t)\right)$ is solution of the ODE~(\ref{eq:mult1}).} Note that the underlying vector field
\begin{equation}
\mathbf{F}\left(y_1,\ldots,y_M\right)=\left(F_1\left(y_1,\ldots,y_M\right),\ldots,F_M\left(y_1,\ldots,y_M\right)\right)
\end{equation}
in equation~(\ref{eq:mult1}) is differentiable and thus, $\mathbf{F}$ is locally Lipschitz. Therefore, solution of~(\ref{eq:mult1}) exists locally and it is unique. Since the state space of interest $\left[0,1\right]^{M}$ is compact and invariant, $\mathbf{F}$ is globally Lipschitz over $\left[0,1\right]^{M}$, and any solution of~(\ref{eq:mult1}) is defined for all time $t\geq 0$ and is unique. In particular, any convergent subsequence converges to the same weak limit given by the unique solution of~(\ref{eq:mult1}) and, thus, the whole sequence converges. This concludes the proof of Theorem~\ref{th:convmulti}.
\end{proof}

\subsection{Multivirus over Multipartite Networks}
\label{subsec:multivirus}

We denote as $\left(\mathbf{\overline{Y}}^{\mathbf{N}}(t)\right)=\left[\overline{Y}_{ik}^{\mathbf{N}}(t)\right]_{ik}$ the matrix process collecting the fraction of $k$-infected nodes at island $i \in \left\{1,\ldots,M\right\}$ with $k \in \left\{1,\ldots,K\right\}$ over time $t$, $t\geq 0$, where $K$ is the number of virus strains. In this Subsection, we refer to $\left(\mathbf{\overline{Y}}_{i}^{\mathbf{N}}(t)\right)=\left(\overline{Y}_{i1}^{\mathbf{N}}(t),\ldots,\overline{Y}_{iK}^{\mathbf{N}}(t)\right)$ as the distribution of infected nodes at island $i$ across the $K$ strains of virus. Recall the definition of the state-space~$\mathcal{Y}^{\mathbf{N}}_i$ of the process $\left(\mathbf{\overline{Y}}_{i}^{\mathbf{N}}(t)\right)$ given in Section~\ref{sec:notation}. Applying the same reasoning as in Section~\ref{sec:pointprocesses}, we obtain the following pathwise description for $\left(\mathbf{\overline{Y}}^{\mathbf{N}}(t)\right)$

\begin{equation}
Y_{ik}^{\mathbf{N}}(t) =  Y_{ik}^{\mathbf{N}}(0)+M_{ik}^{\mathbf{N}}(t)+\sum_{j\sim i}\int_{0}^{t}\gamma^{k}_{ji} Y_{jk}^{\mathbf{N}}(s-)\left(\frac{N_i- \langle \mathbf{Y}_{i}^{\mathbf{N}}(s-),\mathbf{1}\rangle}{N_i}\right)ds-\int_{0}^{t}\mu_i^k Y_{ik}^{\mathbf{N}}(s-)ds,\nonumber
\end{equation}
where the martingale is given by
\begin{eqnarray*}
M_{ik}^{\mathbf{N}}(t)&=& \sum_{j\sim i}\sum_{y\in\mathcal{Y}^{\mathbf{N}}_i}\sum_{q = 1}^{N_j} \underbrace{\int_{0}^{t} {\bf 1}_{\left\{Y_{jk}^{\mathbf{N}}(s-)=q, \mathbf{Y}_{i}^{\mathbf{N}}(s-)=y\right\}}\left(\mathcal{N}_{\gamma^k_{ji} q\left(\frac{N_i-\langle y,\mathbf{1}\rangle}{N_i}\right)}(ds)-\gamma_{ji} q\left(\frac{N_i-\langle y,1\rangle}{N_i}\right)ds\right)}_{M^{\mathbf{N}}_{j,y,q}(t)}\\
&&-\underbrace{\sum_{\ell= 1}^{N_i} \int_{0}^{t} {\bf 1}_{\left\{Y_{ik}^{\mathbf{N}}(s-)=\ell\right\}}\left(\mathcal{N}_{\ell\mu_i}(ds)-\ell\mu_i ds\right)}_{M_b^{\mathbf{N}}(t)}
\end{eqnarray*}
and from the construction in Section~\ref{sec:pointprocesses}, $\left\{\mathcal{N}_{\gamma_{ji} q\left(\frac{N_i-\langle y,\mathbf{1}\rangle}{N_i}\right)}\right\}_{i,j,q,y}$ is a family of independent Poisson processes.

\begin{theorem}\label{th:convmulti-b}
Let $\left(\mathbf{\overline{Y}}^{\mathbf{N}}(0)\right)\Rightarrow \mathbf{y}_0\in \mathbb{R}^{M\times K}$ with $\mathbf{N}=\left(N_1,\ldots,N_M\right)\rightarrow \infty$ and $\frac{N_i}{N_j}\rightarrow \alpha_{ij}<\infty$ for all $i\sim j$. Then, the normalized sequence $\left(\overline{\mathbf{Y}}^{\mathbf{N}}(t)\right)$ converges weakly to the solution $\left(\mathbf{y}(t,\mathbf{y}_0)\right)$ of the following ODE:
\begin{equation}
\frac{d}{dt}y_{ik}(t) =  \underbrace{\left(\sum_{j\sim i}\overline{\gamma}^k_{ji}y_{jk}(t)\right)\left(1-\sum_{m=1}^{K}y_{im}(t)\right)-\mu^k_i y_{ik}(t)}_{F_{ik}\left(\left[y_{mn}(t)\right]_{mn}\right)}\label{eq:multiva33}
\end{equation}
with $\overline{\gamma}_{ji}^{k}:=\alpha_{ji}\gamma_{ji}^{k}$.
\end{theorem}

\begin{proof}

\textbf{$\left.\mbox{i}\right)$ Martingale vanishes.} Since the underlying Poisson point processes $\mathcal{N}_{\gamma_{ji} q\left(\frac{N_i-\langle y,\mathbf{1}\rangle}{N_i}\right)}$ are independent then, it follows as a Corollary to Theorem~\ref{th:ortho} that the compensated martingales $M^{\mathbf{N}}_{j,y,q}(t)$, are pairwise orthogonal. It turns out that
\begin{eqnarray*}
{\sf E}\left( \sum_{j\sim i}\sum_{y\in\mathcal{Y}^{\mathbf{N}}_i}\sum_{q = 1}^{N_j} M^{\mathbf{N}}_{j,y,q}(t)\right)^2
&\!\!\!\!\!\!\!\!\!=\!\!\!\!&\sum_{j\sim i}\sum_{y\in\mathcal{Y}^{\mathbf{N}}_i}\sum_{q = 1}^{N_j}{\sf E}\left(M^{\mathbf{N}}_{j,y,q}(t)\right)^2\nonumber\\
&\!\!\!\!\!\!\!\!\!=\!\!\!\!&\sum_{j\sim i}\sum_{y\in\mathcal{Y}^{\mathbf{N}}_i}\sum_{q = 1}^{N_j}{\sf E}\left(\int_{0}^{t} {\bf 1}_{\left\{Y_{jk}^{\mathbf{N}}(s-)=q, \mathbf{Y}_{i}^{\mathbf{N}}(s-)=y\right\}}\left(\gamma^k_{ji} q\left(\frac{N_i-\langle y,1\rangle}{N_i}\right)ds\right)\right)\nonumber\\
&\!\!\!\!\!\!\!\!\!\leq\!\!\!\!&\sum_{j\sim i}{\sf E}\left(\int_{0}^{t}\sum_{y\in\mathcal{Y}^{\mathbf{N}}_i}\sum_{q = 1}^{N_j}{\bf 1}_{\left\{Y_{jk}^{\mathbf{N}}(s-)=q, \mathbf{Y}_{i}^{\mathbf{N}}(s-)=y\right\}}\gamma^k_{ji} N_jds\right)\nonumber\\
&\!\!\!\!\!\!\!\!\!\leq\!\!\!\!&M\max_{j=1,\cdots,M}\left\{\gamma^k_{ji} N_j\right\}t,\nonumber
\end{eqnarray*}
and thus,
\begin{equation}
{\sf E}\left( \frac{1}{N_i}\sum_{j\sim i}\sum_{y\in\mathcal{Y}^{\mathbf{N}}_i}\sum_{q = 1}^{N_j} M^{\mathbf{N}}_{j,y,q}(t)\right)^2\leq \frac{M}{N_i}\left(\frac{\max_{j=1,\cdots,M}\left\{\gamma^k_{ji}N_j\right\}}{N_i}\right)t\rightarrow 0.\nonumber
\end{equation}
Similarly ${\sf E}\left(\overline{M}_{b}^{\mathbf{N}}(t)\right)^2\rightarrow 0$ and therefore, the normalized martingale $\overline{M}_{ik}^{\mathbf{N}}(t)$ converges to zero in $\mathcal{L}_2$ for all time $t$, $t\geq 0$. Now, from Doob's inequality and from Theorem~\ref{th:probimplyweak}, we conclude that $\left(\overline{M}_{ik}^{\mathbf{N}}(t)\right)\Rightarrow 0$.

\textbf{$\left.\mbox{ii}\right)$ The family $\left(\mathbf{\overline{Y}}^{\mathbf{N}}(t)\right)$ is tight.}
As in equations~(\ref{eq:variation1})-(\ref{eq:variation5}), we have
\begin{eqnarray}
\label{eq:multivariation-b}
\omega\left(\overline{Y}_{ik}^{\mathbf{N}}, \delta, T\right)&\leq& \sup_{0\leq t\leq t+\delta\leq T}\left\{\sup_{u,v\in\left[t,t+\delta\right]}\left|\overline{M}_{ik}^{\mathbf{N}}(u)-\overline{M}_{ik}^{\mathbf{N}}(v)\right|\right\}
+M\frac{\max_{j=1,\cdots,M}\left\{\gamma^k_{ji}N_j\right\}}{N_i}\delta\nonumber\
\end{eqnarray}
From Theorem~\ref{th:arzela} and similar arguments as in Section~\ref{sec:meanfield}, we conclude that $\left(\mathbf{\overline{Y}}^{\mathbf{N}}(t)\right)$ is a tight family, that is, it admits a convergent subsequence $\left(\mathbf{\overline{Y}}^{\mathbf{N_k}}(t)\right)\Rightarrow \left(\mathbf{\overline{Y}}(t)\right)$. Also, from Theorem~\ref{th:arzela}, $\left(\mathbf{\overline{Y}}(t)\right)$ is almost surely continuous.

\textbf{$\left.\mbox{iii}\right)$ If $\left(\mathbf{\overline{Y}}^{\mathbf{N}_k}(t)\right)\Rightarrow \left(\mathbf{\overline{Y}}(t)\right)$ then, $\left(\mathbf{\overline{Y}}(t)\right)$ is solution of the ODE~(\ref{eq:multiva33}).} It follows similarly to as done in the proof of Theorem~(\ref{eq:finalconvergence}), remarking that we assume a finite (fixed) number of islands $M$.

\textbf{$\left.\mbox{iv}\right)$ $\left(\mathbf{\overline{Y}}^{\mathbf{N}}(t)\right)\Rightarrow \left(\mathbf{\overline{Y}}(t)\right)$, where $\left(\mathbf{\overline{Y}}(t)\right)$ is solution of the ODE~(\ref{eq:multiva33}).} Note that for similar reasons as exposed in Subsection~\ref{subsec:singlevirus}, solution of~\eqref{eq:multiva33} exists and is unique. Thus, any convergent subsequence $\left(\mathbf{\overline{Y}}^{\mathbf{N_l}}(t)\right)$ converges to the same weak limit given by the unique solution of~(\ref{eq:multiva33}) and thus, the whole sequence converges.
\end{proof}

\section{Conclusion}\label{sec:conclusion}

In this paper, we established the fluid limit dynamics of a multivirus epidemics over a multipartite network from a peer-to-peer stochastic network model of diffusion. Namely, we proved that the normalized macrostate $\left(\mathbf{\overline{Y}}_{ij}^{\mathbf{N}}(t)\right)$ collecting the fraction of $j$-infected nodes $\left(\overline{Y}^{\mathbf{N}}_{ij}(t)\right)$ per island $i \in \left\{1,\ldots,M\right\}$ with $j\in \left\{1,\ldots,K\right\}$ over $G^{\mathbf{N}}$ converges weakly, under the Skorokhod topology on the space of \emph{c\`{a}dl\`{a}g} sample paths, to the solution $(y(t))$ of a $\left(M\times K\right)$-dimensional ordinary differential equation given by~(\ref{eq:multiva33}). To this effect, we first proved that the underlying martingale perturbation $\left(\overline{M}^{\mathbf{N}}(t)\right)$ vanishes as $\mathbf{N}$ grows large, which implies that the macrostate family $\left(\mathbf{\overline{Y}}_{ij}^{\mathbf{N}}(t)\right)$ is tight in $\mathbf{N}$. Then, we showed that any weak accumulation point of the family $\left(\mathbf{\overline{Y}}^{\mathbf{N}}(t)\right)$ is solution to the vector ordinary differential equation~(\ref{eq:multiva33}) with Lipschitz vector field. From the uniqueness of the solutions of the resulting meanfield differential equation~(\ref{eq:multiva33}), we concluded that the whole sequence $\left(\mathbf{\overline{Y}}_{ij}^{\mathbf{N}}(t)\right)$ converges weakly to the solution of~(\ref{eq:multiva33}). We now present a numerical experiment of two strains of virus $x$ and $y$ spreading across a bipartite network via our SIS stochastic peer-to-peer law of infection. Figure~\ref{fig:completetwo} illustrates the Matlab results for the evolution of two strains in the bipartite network (refer to the noisy curves) and we superimpose on it the corresponding meanfield evolution (refer to the smooth curves). Figures~\ref{fig:complete100}, \ref{fig:complete1000}, and~\ref{fig:complete4000} illustrate the evolution of the fractions of $x$-infected (blue/solid curves) and $y$-infected (red/dashed curves) nodes at islands $1$ and $2$. The boldfaced curves represent the associated meanfield solutions.
\begin{figure}[htb]
        \centering
        \begin{subfigure}[htb]{0.3\textwidth}
                \centering
                \includegraphics[width=\textwidth]{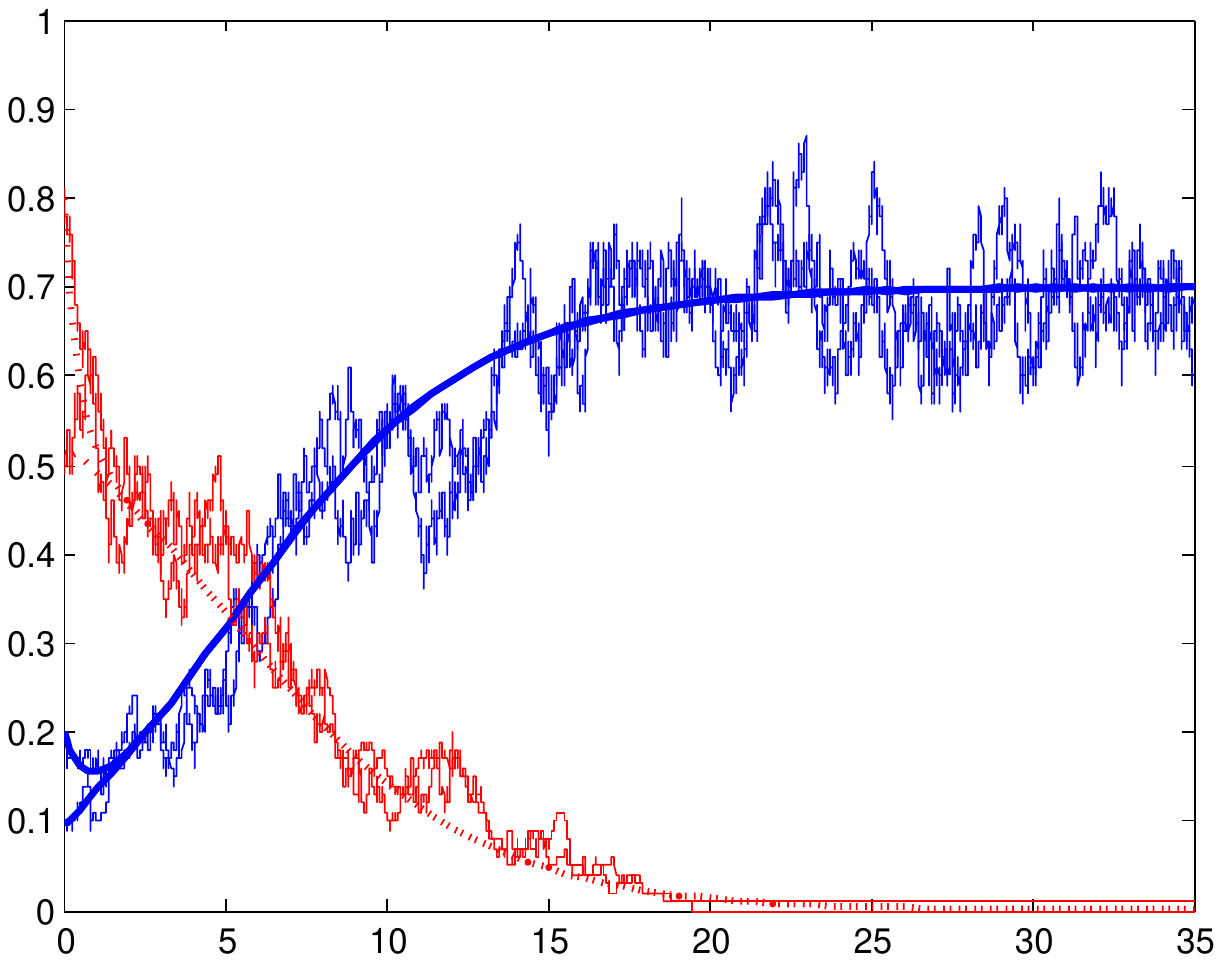}
                \caption{100 nodes per island.}
                \label{fig:complete100}
        \end{subfigure}%
        \begin{subfigure}[htb]{0.3\textwidth}
                \centering
                \includegraphics[width=\textwidth]{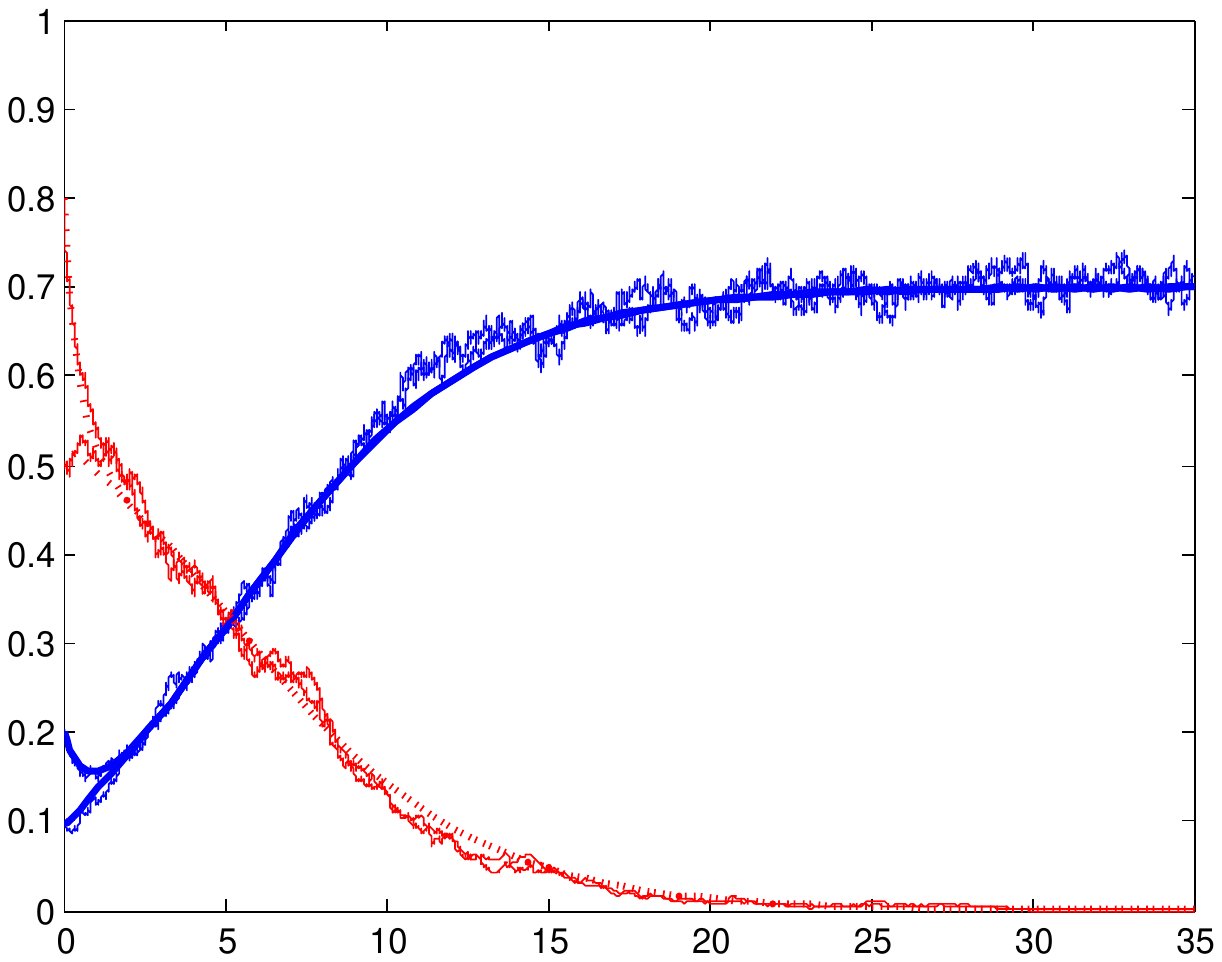}
                \caption{1000 nodes per island.}
                \label{fig:complete1000}
        \end{subfigure}
        \begin{subfigure}[htb]{0.3\textwidth}
                \centering
                \includegraphics[width=\textwidth]{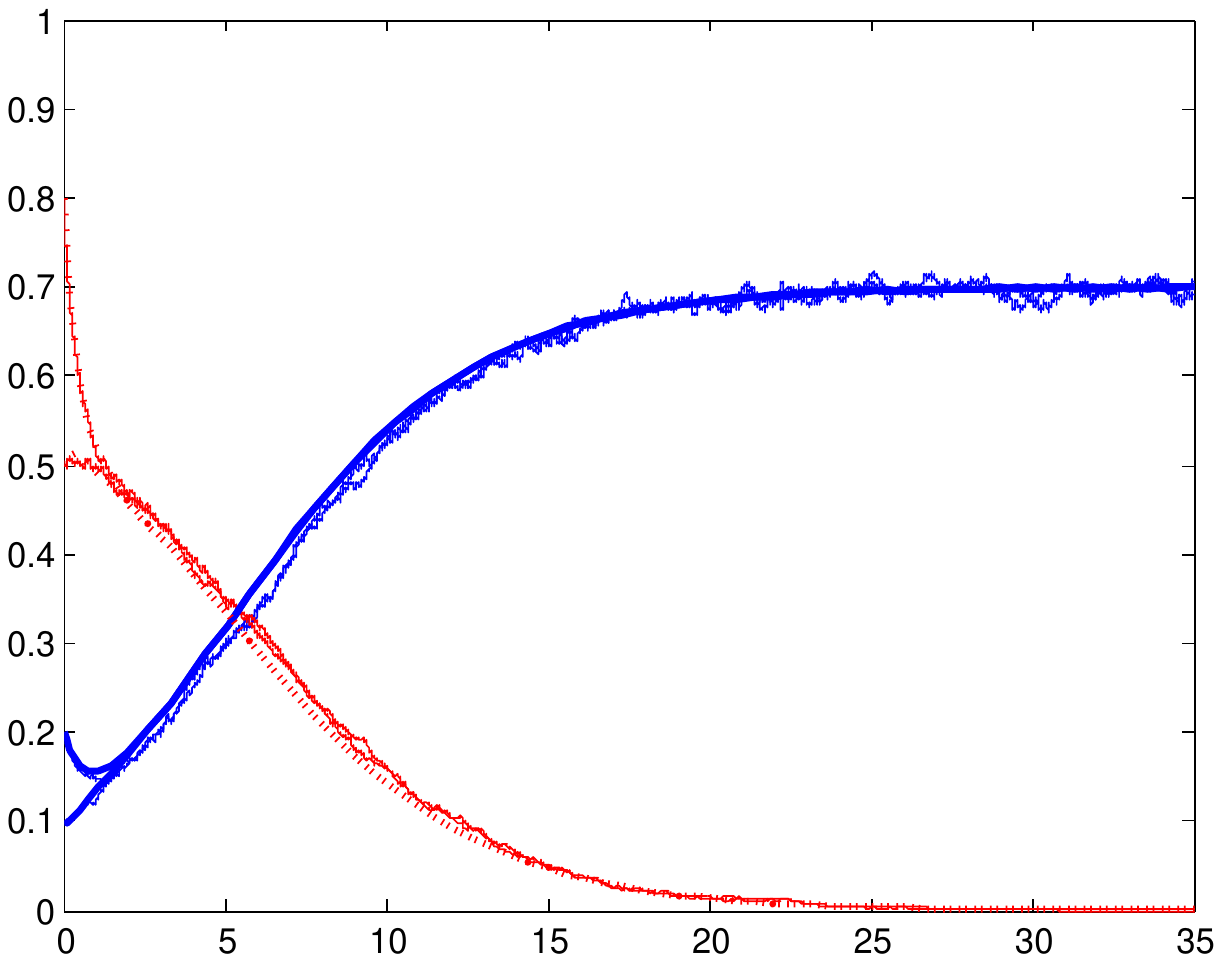}
                \caption{4000 nodes per island.}
                \label{fig:complete4000}
        \end{subfigure}
        \caption{The plots represent the numerical evolution of the fractions of $x$-infected (in blue or solid) and $y$-infected (in red or dashed) nodes at each island $1$ and $2$. The boldfaced curves represent the solution of the limiting vector differential equation of a bi-viral epidemics in a bipartite network.}\label{fig:completetwo}
\end{figure}
We observe that, as the number of nodes grows large, the randomness decreases and the infected population dynamics fits the meanfield prediction. For the experiment, we have set $\gamma^x:=\gamma_{12}^x=\gamma_{21}^x> \gamma_{12}^y=\gamma_{21}^y=:\gamma^y$ with $\mu^x=\mu^y=1$ and observe, in particular, that the most virulent strain survives and the weaker strain dies out. The qualitative analysis of the meanfield dynamics~(\ref{eq:multiva33}) is developed in~\cite{Qualitative} where we proved that the natural selection phenomenon observed in Figures~\ref{fig:complete100}, \ref{fig:complete1000}, and~\ref{fig:complete4000} extends to symmetric regular multipartite networks, where by symmetric and regular we mean $\gamma^y=\gamma^y_{ij}$ and $d=d_i=d_j$ for all pair of communicating islands $i$ and $j$ and for all types of virus $y$.

\small

\bibliographystyle{IEEEtran}
\bibliography{IEEEabrv,biblio_new}


\end{document}